\documentclass[%
aip,
amsmath,amssymb,
reprint, onecolumn 
]{revtex4-1}

\usepackage{graphicx}
\usepackage{dcolumn}
\usepackage{bm}

\usepackage{amsthm}
\usepackage[utf8]{inputenc}
\usepackage[T1]{fontenc}
\usepackage{mathptmx}
\usepackage{etoolbox}

\usepackage{graphics}
\usepackage{amsmath}
\usepackage{comment}
\usepackage{amssymb}
\usepackage{xcolor}
\usepackage{float}
\usepackage{hyperref, color}
\usepackage[title]{appendix}
\usepackage{amsthm}

\usepackage{verbatim} 
	
	\newtheorem{theorem}{Theorem}[section] 
	\newtheorem{lemma}[theorem]{Lemma} 

	\newtheorem{proposition}{Proposition}
	\newtheorem{corollary}{Corollary}[proposition]
	
	\def\ni{\noindent}
	\def\ph{{\phantom{...}}}
	\def\={\phantom{..} = \phantom{..}}

	\def\+{\phantom{..} + \phantom{..}}
	\def\>{\phantom{..} > \phantom{..}}
	\def\<{\phantom{..} < \phantom{..}}
	\def\-{\phantom{..} - \phantom{..}}

	\def\bq{\begin{quote}}
		\def\eq{\end{quote}}
	\def\be{\begin{equation}}
		\def\ee{\end{equation}}
	\def\bar{\begin{eqnarray}}
		\def\ear{\end{eqnarray}}
	\def\no{\nonumber}
	\def\pa{\partial}

	\def\Gartner{G{\"a}rtner}

	\def\wfs{wavefunctions}

	\def\pN{\prod_{j=1}^N}

	\def\Plb{P\left[\,}
	\def\Prb{\,\right]}
	\def\sumS{{\sum_S\,|\psi|^2(S)}}
	\def\sumi{{\sum_{i=1}^N}}

	\def\sumi{\underset{i=1}{\overset{N}{\sum}}}
	\def\sumk{\underset{k=1}{\overset{N}{\sum}}}
	
	\def\ket{\right\rangle}
	\def\bra{\left\langle}
	\def\verti{\left|\right.}
	
	\newcommand{\dverti}[1]{\left| #1 \right|}
	
	\makeatletter
	\def\@email#1#2{%
		\endgroup
		\patchcmd{\titleblock@produce}
		{\frontmatter@RRAPformat}
		{\frontmatter@RRAPformat{\produce@RRAP{*#1\href{mailto:#2}{#2}}}\frontmatter@RRAPformat}
		{}{}
	}%
	\makeatother
	
\begin{document}
		
		\preprint{AIP/123-QED}

		\title[]{On blocking Dispersion of Matter by Energy conservation}
		\author{Leonardo De Carlo}
		\email{leonardo\_d3\_carlo@protonmail.com}
		\affiliation{ 
			CopeLabs, Universidade Lusófona, Campo Grande 388, Lisboa, Portugal
		}
		\affiliation{
			Grupo de Física Matemática, Inst. Superior Técnico, Av. Rovisco Pais 1, Lisboa, Portugal 
		}%
		
		
		\date{\today}
		
		\begin{abstract}
			In [L. De Carlo and W. D. Wick, Entropy 25, 564 (2023)], we considered the problem of wavefunction ensembles for spin models. To observe magnetization at finite temperature, we had to add macroscopic nonlinear terms that suppress macroscopic superpositions by energy conservation. The nonlinear terms were of the kind introduced in [W. D. Wick, arXiv:1710.03278 (2017)] to block spatial cats by energy conservation, proposed as a solution to the Measurement Problem. Reviewing this theory, I derive commutation relations that these nonlinear terms have to satisfy to be physically admissible, and show that the ones confining the wavefunction in space do indeed satisfy these relations. I generalize the terms of [L. De Carlo and W. D. Wick, Entropy 25, 564 (2023)] for non-pure spin models and check if they also satisfy the constraints, concluding with a negative answer and possibly giving an interpretation of the previous results. With a toy model I present the main experimental idea, that is that forming spatial cats meets an energy barrier. A comparison with Collapse Models is at the end.
			
			\vspace{0.2cm}
		\end{abstract}

		\maketitle

		\section{Introduction}\label{sec:intro}
		
		In recent years creating objects in macroscopic superpositions is becoming an  important experimental research topic \cite{ArnHorn,Q-Tube,Q-Xtreme}. The main quest is if the  laws of quantum physics hold for massive objects of arbitrary size. A spatial macroscopic superposition is an  object delocalized over a distance comparable to its spatial extent.  While Newtonian physics is restored if the spatial dispersion of observed objects stays small on the scale of variation of external potentials. Therefore  a possible goal is that of  locating the classical-quantum boundary. If such boundary exists, one can expect  superposition principle breaking down. Famous models to address the question are collapse models, see a recent review in \cite{Collapse}. Here I present an alternative idea, that is introducing macroscopic terms blocking cat states by energy conservation. 
		For a discussion and critiques of current paradigms on the classical-quantum boundary, see \cite{HanceHossenfelder2022}.
		
		\vspace{0.5cm}
		
		The idea introduced in \cite{MP1}  shares some features with collapse models, that is  adding macroscopic non-linear terms becoming active at large enough scale, but different in the physical framework,  i.e.  macroscopic superpositions  are ``penalized'' by energy conservation. Meaning that  the classical-quantum boundary will  not be only a function of the size   but also of energy.  A comparison between the two models is made in Section \ref{sec:comparison}.

		\vspace{0.5cm}
		
		The adopted principles to modify the Schr\"odigner's equation are :
		\begin{itemize}
			\item[a)] The modification has to be negligible at microscopic level but it becomes large at macroscopic level to block macroscopic dispersion in space because of ``configurational cost'';
			\item[b)] The norm of $ \psi $ and the energy of a closed system have to be conserved ;
			\item[c)] No extra terms are added to the evolution of the center of mass $ X(\psi)= \bra  \psi \dverti{\frac{1}{N}\underset{i=1}{\overset{N}{\sum}}X_i}    \psi \ket   $  of a closed system.
		\end{itemize}

		With respect to \cite{DeCarlo-Wick} I add 
		\begin{itemize}
			\item[d)] The dynamics of a  system, coupled with a microscopic one``to measure'',  has to display chaos.
		\end{itemize}
		
		To maintain concise notation, I consider the one-dimensional case and I introduce what they   called\cite{MP1} \emph{Spatial Dispersion} of a wavefunction: 
		\begin{equation}\label{eq:sD}
			D_X(\psi, \psi^*)= \bra \psi \dverti{\frac{1}{N^2}\left(\,\overset{N}{\underset{i=1}{\sum}}\,X_i - \bra\psi\dverti{\overset{N}{\underset{i=1}{\sum}}\,X_i}\psi\ket\,\right)^2}  \psi\ket.
		\end{equation}
		
		Calling $r$ the spatial extension of an  object (the support of $\phi_r$) and $R$ a comparable or larger distance,  consider two wavefunctions representing the  system $i=1,\dots, N$:
		\begin{equation}\label{eq:MQP}
			\pN\left\{ \frac{1}{\sqrt{2}} \phi_r(x_j+R) + \frac{1}{\sqrt{2}} \phi_r(x_j-R)\right\}
		\end{equation}
		and
		\begin{equation}\label{eq:spcat}
			\frac{1}{\sqrt{2}}\pN \phi_r(x_j+R) + \frac{1}{\sqrt{2}}\pN \phi_r(x_j-R).
		\end{equation}
		\ni The first states represent  a "macroscopic quantum phenomena" (MQP), where many particles are in superposition but they are not a macroscopic superposition, while the second state are  the relevant ones in the quantum theory of measurement \cite{Leggett}, this can be seen plugging \eqref{eq:MQP} and \eqref{eq:spcat} in \eqref{eq:sD}, a simple example is given by using a normalized characteristic function $ \phi_r(x) $. The first ones gives $D_X=O(R^2/N)$, while the second ones gives $D_X=O(R^2)$, namely the spatial dispersion of \eqref{eq:MQP} is small, while for  \eqref{eq:spcat} it is macroscopic: the object has a ``macroscopic dispersion'' forming a spatial cat. Wavefunction \eqref{eq:MQP} has the form of a distribution for i.i.d. random variable, so $D_X=O(R^2/N)$ is a central limit theorem behavior. For  something like $ 	\pN \phi_r(x_j) $ it is $ D_X(\psi)= O\left(\frac{r^2}{N}\right) $.
		The idea to satisfy  a), b) and c)  is exploiting the behavior of \eqref{eq:sD} to forbid the states \eqref{eq:spcat} by energy conservation  rewriting the evolution of $\psi  $ in following the Hamiltonian form 
		\begin{equation}
			\begin{split}
				i\hbar\frac{\partial \psi}{\partial t} &= \frac{\partial}{\partial\psi^* }E(\psi), \text{ with } E(\psi) = E_{QM}(\psi) + E_{WFE}(\psi),\\
				E_{WFE}(\psi) &= w N^2 D_X(\psi, \psi^*)
			\end{split}
		\end{equation}
		where `$ w $' is a very small positive constant and  the term $ E_{WFE}(\psi) $ is named \emph{"Wavefunction-Energy"}(WFE). What they introduced\cite{MP1} is 
		a universal self-trapping to forbid macroscopic dispersion and restore  a ``classical'' behavior.   In Sections \ref{sec:derivationcom} and \ref{sec:HWMeasurment}, I explain how it satisfies a), b) c), deriving some general conditions for  terms like $  w N^2 D_X(\psi, \psi^*) $, which can be tested on other cases. Indeed another interesting case, and arguably  a better physical choice, is found. 

		\vspace{0.5cm}
		
		In \cite{DeCarlo-Wick}  (looking for a wavefunction thermodynamics) we did the ``experiment'' to consider wavefunction ensembles for pure spin models, that is  we evaluated  probability  generating functions of  the form:
		\begin{equation}\label{eq:Z}
			Z_N \=\frac{1}{\Sigma_N} \int_{||\psi|| = 1}\,d\psi\,\exp
			\left\{\,- \beta\,E_N(\psi) 
			\,\right\},
		\end{equation}
		where $||\psi||^2 = \underset{s}{\sum}\,|\psi(s_1,...,s_N)|^2$ and  $ s=\left\{s_1,\dots,s_N\right\} $ a spin configuration . 
		The integral is over the unit sphere, normalized to its volume $\Sigma_N$. The energy is	$E_N(\psi) =  <\psi|\,E(s)\,|\psi>=  \underset{s}{\sum}\,|\psi(s_1,...,s_N)|^2\, E(s),$
		with $ E(s) $ a   spin magnetic energy of order $ N $ in the configuration $ s $, as Curie-Weiss or Ising.  The question was to understand if it was possible to construct models  exhibiting a spontaneous magnetization in the thermodynamics limit  for sufficiently low temperature, i.e. ensembles where 
		\begin{eqnarray}\label{eq:mag}
			\underset{N\rightarrow\infty}{\lim}&[\,m^2(\psi)\,]_{\beta}  
			\= \underset{N\rightarrow\infty}{\lim}\,\, \frac{1}{Z_N} \frac{1}{\Sigma_N}\int_{||\psi|| = 1}\,\,d\psi\,\exp\{ - \beta\,E_N(\psi)\,\}\,m^2(\psi) >0, \text{ for $\beta$ large enough}, \no\\
			&m(\psi) = \bra\psi\dverti{\,\underset{i}{\sum}\,\frac{s_i}{N}}\psi\ket = \underset{s}{\sum}\,|\psi(s)|^2\left(\,\underset{i}{\sum}\,\frac{s_i}{N}\,\right).\no
		\end{eqnarray}
		
		We concluded that, to observe this, two physical ingredients were necessary. The first was to  consider  indistinguishable spin variables, namely imposing exchange symmetry on the wavefunction $\psi(s_1,\dots,s_N)$. This means that the high dimensionality $2.2^N$ of the Hilbert space suppresses such phenomena in the distinguishable spins case.  The second was to introduce a  nonquadratic term, in  addition to the energy $E_N(\psi)$, that penalizes large superpositions in large objects. The term was  of the form
		\begin{equation}\label{eq:genwfe}
			E_{WFE}(\psi)= wN^2 D(\psi),\,\, D(\psi):=  \bra\psi \dverti{\left(\,\frac{1}{N}\overset{N}{\underset{i=1}{\sum}}\,O_i - \bra\psi\dverti{\,\frac{1}{N}\overset{N}{\underset{i=1}{\sum}}\,O_i\,}\psi\ket\,\right)^2} \, \psi \ket,
		\end{equation}
		where  $O_i=s_i$. In appendix \ref{app:proof}, I report  the construction of the type of models having a phase transition at finite temperature and the main theorem proved in \cite{DeCarlo-Wick} with its condensed proof.

		At this point I declare the question I pose here. The natural extension of $ O_i =s_i $ to wavefunctions $ \psi(x)\psi(s) $ is $ O_i=L_i+s_i $, therefore I will check in section \ref{sec:Jcase} if this case satisfies the constraints a),b),c). The conclusion is negative.  These terms were proposed to affect only macroscopic objects (and  otherwise be too small to matter). 
		Since a magnetic field can move a magnetometer needle, it might be considered macroscopic. A possible interpretation is that the non-linear terms added directly to the spin system in \eqref{eq:Z} are a way to evade modeling the magnetometer. Another possibility might be considering a model of magnet with both spatial and spin coordinates, discretizing  wavefunction states on a lattice with space `$a$' in presence of both spin interactions and WFE,  to explore the indirect role of a spatial WFE on the spin coordinates. But this is  a future  research program.

		\vspace{0.5cm}
		The problem of wavefunction ensembles was introduced in \cite{ES1,ES2,Bloch} and has been revisited in recent time in \cite{Campisi,Jona-Presilla,lebowitz,Anza}. In particular, in \cite{Brody} they concluded that it would have been interesting  exploring non-linear modification of Schr\"odinger's evolution in such nonconventional ensembles.

		\vspace{0.5cm}

		In Section \ref{sec:HWMeasurment}, I describe an ideal experiment to test WFE and and show some rough orders of magnitude of $w$ using common orders of magnitudes .  Given the lack of the realization of this ideal experiment,  I can only describe qualitatively what one should expect to observe.  In Section \ref{sec:HWMeasurment},  I apply to a toy model some analytical conditions, developed in \cite{MP3,ContinuumChaos}, that are sufficient to establish the existence of expanding and contracting directions at a given threshold. A final comparison  with Collapse models is made in Section  \ref{sec:comparison}. 
		
		\vspace{0.5cm}

		This document wants to present previous work self-consistently to interpret the results of \cite{DeCarlo-Wick}, trying to give some estimates for  the order of the coupling constant $w$, and provide a comparison with Collapse models. The reader should distinguish between material already present in the literature (Sections \ref{sec:intro}, \ref{sec:derivationcom}, and the toy model in Section \ref{sec:HWMeasurment}) and the new additions provided in this work. These additions are the commutation relations, Section \ref{sec:Jcase}, the latter part of Section \ref{sec:HWMeasurment} (after Eq. \eqref{eq:DCI}), and Section \ref{sec:comparison}.

		\vspace{0.5cm}
		
		\textbf{Abbreviations and Notation.} The following abbreviations are used in this manuscript:\\
		
		\begin{tabular}{@{}ll}
			
			COM & center-of-mass\\
			MQP & macroscopic quantum phenomena\\
			WFE & Wavefunction Energy\\
			X-WFE &  Wavefunction energy on Position\\
			P-WFE & Wavefunction energy on Momentum 
			\\
			The symbol $ \langle \,\,\, |\,\,\,\rangle$ denote the scalar product in a $ L_2(\Omega)$-space. 
		\end{tabular}

		\section{ The proposal of blocking cat states by energy conservation}\label{sec:derivationcom}

		\subsection{Hamiltonian macroscopic modifications} \label{sec:ModwithWFE}
		
		Here, along the line of \cite{MP1}, I  explain how    properties a), b) and c) are satisfied.   Consider the  evolution of the  state $ \psi $ given by 
		\begin{equation}\label{eq:nonlin2}
			i\hbar\frac{\partial \psi}{\partial t}=\frac{\partial}{\partial\psi^* }E(\psi),\text{ with } E(\psi) \= E_{QM}(\psi) \+ E_{WFE}(\psi),
		\end{equation}
		where $ w $ is a very small positive constant. The $ O_i $ will have to be self-adjoint operators such that $\left(\overset{N}{\underset{i=1}{\sum}}\,O_i  \right)^2$ is self-adjoint too. (In \cite{DeCarlo-Wick} we made  a typo in the final printing writing ``where $w$ is a very small constant and the $O_i$ ’s are self-adjoint operators diagonal in the
		same base as $H_{QM}$ ''. This was the correct statement.) To have c),  two  restricting conditions  on some commutators between the operators  $X_k, i\frac{\partial}{\partial x_k}$ of the  particle $k$ and the family $\{O_i\}_i$ need to be imposed. These conditions  seem to restrict the choice between $O_i=X_i$ and $O_i=P_i$.   Each $O_i$ acts on the $i$-argument of the wavefunction $\psi(x_1,\dots,x_n)$. 

		\subsection{Properties of Wavefunction Energy}
		
		\begin{itemize}
			\item[\textbf{a)}]  It follows from having  $w$ very small  and  an energy scaling as $N^2$ on spatial cat states,  for example taking $O_i=X_i$ as described after \eqref{eq:sD}. Given an initial state $\psi$ of dispersion such that  $E_{WFE}(\psi)= wN^2D_X(\psi) \sim O(wN)$,  when $N$ becomes large, \eqref{eq:nonlin2}   can not evolve into \eqref{eq:spcat} since it is a Hamiltonian evolution (see \eqref{eq:hamevo}) and energy must be conserved: on cat states $ E_{WFE}(\psi)= wN^2D_X(\psi) $ scales with order  $N^2$, therefore they  will become rapidly too expensive to be created. To visualize this, let's consider some illustrative orders of magnitude: with $w \sim 10^{-25}\,\, J/m^2$, $N\sim 10^{20}$ and $R\sim 1 \,\,cm$ , for a cat we would have $E_{WFE} \sim wN^2R^2=10^{11} J  $, while for an initial product state \eqref{eq:MQP}  with  $R'= 1\mu m$ we would have $E_{WFE}=wR'^2N= 10^{-17}J$. So the energy for a cat can not be supplied by the initial state itself, unless enough energy is supplied from an external potential. While for an hydrogen atom $E_{WFE}\sim 10^{-25} (6 \cdot 10^{-11})^2\sim 10^{-47} J$. 
			
			This rapid threshold is at the core of the estimate proposed in section \ref{sec:HWMeasurment}, starting from  an illustrative toy model.  
			
			\item[\textbf{b)}]  I need to consider $\mathbf{E}:=E(\psi,\psi^*)$ as functional of $\psi$ and $\psi^*$ to observe  the symplectic structure of \eqref{eq:nonlin2}. Defining the coordinates $Q:= (\psi+\psi^*)/2$ and $P:=(\psi- \psi^*)/2i$.  I have an Hamiltonian system where $E(\psi,\psi^*)=E(Q,P)$.  The evolution becomes 
			\begin{equation}\label{eq:hamevo}
				(i\hbar)\frac{d\,\,F(\psi,\psi^*)}{dt}= \{F,E\}_W(\psi,\psi^*),\,\,\left\{ F,  G\right\}_W(\psi,\psi^*)= \left\langle  \frac{\partial F}{\partial \psi} \frac{\partial G}{\partial \psi^*} -   \frac{\partial G}{\partial \psi} \frac{\partial F}{\partial \psi^*} \right\rangle,
			\end{equation}
			for any pair of functional $F(\psi, \psi^*)$ and $G(\psi, \psi^*)$, where  `W' stays for wavefunction Poisson bracket $\bra{\cdot, \cdot}\ket$, replacing the usual Poisson bracket after observing  $\{\cdot,\cdot \}=\frac{2}{i}\{\cdot,\cdot\}_W$ .  
			Respectively  $ \bra \cdot \ket $ and $\frac{\partial}{\partial\psi}$ denote an integration and a functional derivative  for continuous arguments and a sum  and derivative for   discrete arguments. One gets the commutator of quantum mechanics when $F(\psi,\psi^*)$ and $G(\psi,\psi^*)$ are quadratic, i.e. $\langle \psi |F|\psi \rangle$     and $\langle \psi |G|\psi \rangle$, defined by $F$ and $G$ self-adjoint: $\left\{F(\psi), G(\psi)\right\}_W= \langle \psi |[F,G]|\psi \rangle$, where $[F,G]=FG-GF$. 
			
			From \eqref{eq:hamevo} it follows $(i\hbar)\frac{d\,\,E(\psi,\psi^*)}{dt}= \{E,E\}(\psi)=0$. Next I  show that $\|\psi\|$ is conserved.
			Taking $F(\psi)=\|\psi\|^2$, from \eqref{eq:hamevo} I have
			\begin{equation*}
				\frac{\partial}{\partial t } \left\langle \psi \middle| \psi \right\rangle = \frac{1}{i\hbar}\left[ \left\langle \psi \middle| \frac{\partial \mathbf{E}}{\partial\psi^*} \right\rangle-\left\langle  \frac{\partial \mathbf{E}}{\partial\psi} \middle|\psi \right\rangle\right],
			\end{equation*}
			this is zero for any $\psi$  if  
			\begin{equation}\label{eq:zerocond}
				\frac{\partial\mathbf{E}_{WFE}}{\partial\psi^*}= \mathcal{O}(O_1,\dots,O_N;\psi,\psi^*)\psi,
			\end{equation}
			with $\mathcal{O}(O_1,\dots,O_N;\psi,\psi^*)$ self-adjoint operator possibly depending on $\psi,\psi^*$ in some fashion. Now I verify this condition for \eqref{eq:genwfe}. I have:
			\begin{equation}\label{eq:WFEder}
				\frac{\partial\mathbf{E}_{WFE}}{\partial\psi^*}= \left\{ w\left(\overset{N}{\underset{i=1}{\sum}}\,O_i\,\right)^2 - 2 w \left\langle\psi\left|\,\overset{N}{\underset{i=1}{\sum}}\,O_i\,\right|\psi\right\rangle \left(\,\overset{N}{\underset{i=1}{\sum}}\,O_i \right)\right\}\psi
			\end{equation}
			the term inside curly bracket satisfies \eqref{eq:zerocond} when $O_i$ are self-adjoint and $\left(\overset{N}{\underset{i=1}{\sum}}\,O_i\,\right)^2$ too.
			
			\item[\textbf{c)}] The equation of the center of mass(COM) has to remain unchanged, meaning no additional terms are added to Newton's equation. This imposes extra constraints on the $O_i$'s, which are expressed by certain commutators involving $X_k$ and $-i\partial_k$ of $ k$-th particle.
			I show the computations for  the one-dimensional case, but everything generalizes  introducing more notation.   I consider a microscopic system of coordinate $y$ entangled with a ``macroscopic'' system made by $N$ particle $(x_1, \dots, x_N)$, to keep in mind situations of interest in  the Measurement Problem, where one could  expect the whole system evolving into macroscopic superpositions. I call `$V$' the potential energies of the large system and `$U$' the interaction energy with the microscopic system $y$.    For convenience sometimes I will use $ F(\psi):= F(\psi,\psi^*)$.
			
			First,  from the linear part of \eqref{eq:nonlin2}
			\begin{equation}\label{eq:ScheqManyB}
				i\hbar\frac{\partial\psi}{\partial t}= - \left(\frac{\hbar^2}{2m}\right)\underset{i=1}{\overset{N}{\sum}} \Delta_i \psi -\left(\frac{\hbar^2}{2m}\right) \Delta_y \psi + V(x_1, \dots, x_N)\psi + U(x_1,\dots,x_N, y )\psi := H\psi.
			\end{equation} 
			
			I want to derive $\ddot{X}(\psi)=\frac{d^2}{dt^2}\left\langle\psi\left|\frac{1}{N}\overset{N}{\underset{i=1}{\sum}}\,X_i\right|\psi\right\rangle$. For practical  purpose I will use $X$ in place of $X(\psi)$ and $x$ in place of $\frac{1}{N}\sum_{i=1}^{N}x_i$.  I consider $U= \overset{N}{\underset{i\neq j}{\sum}}u(x_i-x_j)+ \tilde{u}\left(X,y\right)$ and $V=\underset{i=1}{\overset{N}{\sum}} v(x_i) $, where $\tilde{u}$ describe the interaction between the COM and the entagled particle. Following the steps of Appendix \ref{app:Newton} we get 
			\begin{equation}\label{eq:finalNew}
				M\ddot{X} = - \sumk \bra \psi | \pa_k v (x_k) |\psi\ket - \sumk \bra \psi | \pa_k \tilde{u} (X,y) |\psi\ket  .
			\end{equation}
			
			First note that $\hbar$ cancels out, meaning I don't relate a classical motion to $\hbar\rightarrow 0$, instead I observe it happens in  \eqref{eq:finalNew}  when the following approximation holds: 
			\begin{equation}\label{eq:clasapprox}
				N\left(- \frac{1}{N}\sumk \bra \psi | \pa_k v (x_k) |\psi\ket\right) \approx - N \partial_{\text{x}} v(X) := - \partial_{\text{x}} V_E(X),  
			\end{equation}
			quoting \cite{MP1} "\emph{the factor of N on the right is absorbed into a “macroscopic”
				external potential energy, rendering it extensive}".
			Namely up to negligible errors, in the LHS of \eqref{eq:clasapprox}, the average sum over  $k$ and the integral over space can be pushed inside the $\partial_k v$.  The first simplification is also present  in classical physics $\frac{1}{N}\sum_k \partial_k v(x_k)\approx \partial_{\text{x}}v(x)$,  and it holds when $\pa_{\text{x}} v$  varies little on the scale of the object. Meanwhile, the second one is purely wave-mechanical, consisting  in $\bra  \psi \dverti{\partial_{\text{x}} v (x_i)}\psi \ket \approx  \pa_{\text{x}} v(\bra\psi\dverti{x_i}\psi\ket)\approx  \pa_{\text{x}} v(\bra\psi\dverti{x}\psi\ket)$. This can happen if the spatial dispersion \eqref{eq:sD} of the observed objects is very small on the scale of variation of external potentials.
			Consequently,  during the dynamics  \eqref{eq:hamevo}  the quantity \eqref{eq:sD} has to maintain small.

			To complete part c), I still need to examine the conditions on \eqref{eq:genwfe} and identify some interesting cases that satisfy it. Again, I study $F(\psi,\psi^*)=\bra\psi \dverti{X_k}\psi \ket $ and look at 
			\begin{equation*}
				\frac{d}{dt}\bra\psi |X_k|\psi \ket = \left[- \frac{i}{\hbar}\right]\bra \partial_\psi F \partial_{\psi^*}E - \partial_{\psi}E\partial_{\psi^*} F\ket, \text{ where } E(\psi,\psi^*)=  \bra \psi \dverti{\left(\,\overset{N}{\underset{i=1}{\sum}}\,O_i\right)^2}  \psi\ket - \left(\bra\psi \dverti{\overset{N}{\underset{i=1}{\sum}}\,O_i}\psi\ket\,\right)^2. 
			\end{equation*}
			Introducing $O_i(\psi):= \bra \psi \dverti{O_i}\psi \ket $, I have
			\begin{align}
				\langle \partial_\psi F \partial_{\psi^*}E - \partial_{\psi}E\partial_{\psi^*} F\rangle &= \bra \psi \left| X_k \left[\left(\sum_{i=1}^N O_i   \right)^2 - 2 \left(\sum_{i=1}^N O_i(\psi) \right)\sum_{i=1}^N O_i \right]\right|\psi\ket \label{eq:cancellation}\\
				&\quad - \bra \psi \left| \left[\left(\sum_{i=1}^N O_i   \right)^2 - 2 \left(\sum_{i=1}^N O_i(\psi) \right)\sum_{i=1}^N O_i \right] X_k \right|\psi\ket  = A - B. \nonumber
			\end{align}
			
			I need that $A=B$. Which means that $X_k$ and $\left(\sumi O_i   \right)^2 - 2 \left(\sumi O_i(\psi) \right)\sumi O_i$ have to commute in expectation. Clearly this will be true for $O_i= X_i$ as in \eqref{eq:sD}, but I want to find some specific conditions for  general $O_i$,  not necessary commuting with $X_i$.  This allows us to find probably a more subtle case.  From $\left(\sumi O_i \right)^2 $ I have the problematics terms $O_k^2 + 2 \underset{j:j\neq k }{\sum} O_j O_k $ and from $\sumi O_i$ just $O_k$  ($O_i(\psi)\in \mathbb{R}$). 
			So $A-B=0$ reduces to asking
			\begin{equation}\label{eq:Magicanc}
				\bra \psi \dverti{X_k G(O_1,\dots,O_N) - G(O_1,\dots,O_N) X_k}	  \psi \ket=0, \text{ where } G(O_1,\dots,O_N) =  O^2_k  +2 \underset{j:j\neq k}{\sum}O_jO_k  - 2 \left(\sumi O_i(\psi) \right) O_k.
			\end{equation}
			Cancellation \eqref{eq:Magicanc} is also satisfied by the relevant example $O_i=P_i$, indeed,  using the commutators $[X_k,P_k]= i\hbar$ and  $[X_k,P_k^2]=2i\hbar P_k$, it simplifies to:
			\begin{equation*}
				2i\hbar\left( P_k(\psi)+ \underset{j:j\neq k}{\sum}P_j(\psi) -  \sumi P_i(\psi)\right)=0.
			\end{equation*}
			
			I find this interesting since the non-trivial cancellation \eqref{eq:Magicanc} regard a commutator between  $X_k$ (which does not commute with $P_k$ and its powers) and a non-trivial operator $G(P_1, \dots, P_k,\dots,P_N)$. At the same time I have an expression for \eqref{eq:genwfe} that does not move the COM and  forbids  a macroscopic object to become a momentum cat,  potentially avoiding a priori a spatial cat . To have an idea one can repeat the computations for \eqref{eq:MQP} and \eqref{eq:spcat} in the momentum representation with $P_i$ in place of $X_i$ in \eqref{eq:sD}. I will discuss this more at a later time.  Now I have still to look at 
			\begin{equation*}
				\frac{d^2}{dt^2}X_k(\psi) = \frac{d}{dt} \frac{i\hbar}{m} \bra \pa_k \psi \vert\psi \ket,
			\end{equation*}
			therefore  in \eqref{eq:hamevo} I consider $F(\psi)=i  \hbar\bra \pa_k \psi \vert\psi\ket $. This time one finds 
			{ \begin{eqnarray}\label{eq:cancellation3}
					&\bra \partial_\psi F \partial_{\psi^*}E - \partial_{\psi}E\partial_{\psi^*} F\ket =	\nonumber\\
					&\bra \psi \dverti{ (-i\hbar\pa_k) \left[\left(\sumi O_i   \right)^2 - 2 \left(\sumi O_i(\psi) \right)\sumi O_i \right]}\psi\ket - \bra \psi \dverti{\left[\left(\sumi O_i   \right)^2 - 2 \left(\sumi O_i(\psi) \right)\sumi O_i \right] (-i\hbar\pa_k) }\psi\ket \nonumber\\
					& = A-B\nonumber.
			\end{eqnarray}}
			I have the same of \eqref{eq:Magicanc} with $(-i\hbar \pa_k)$ in place of $X_k$, giving as condition \begin{equation}\label{eq:Magicanc2}
				\bra \psi \dverti{(-i\hbar \pa_k) G(O_1,\dots,O_N) - G(O_1,\dots,O_N) (-i\hbar\pa_k)}\psi \ket=0.
			\end{equation}
			This condition is verified easily for $O_i=P_i$, but also  for $O_i=X_i$ with similar computations  to the previous ones for $P_i$ in \eqref{eq:Magicanc}.

		\end{itemize}

		I found two cases,  one ($O_i=X_i$) directly restricting the wavefunction in space and another one ($O_i=P_i$) restricting   the wavefunction in momentum, in the sense of forbidding  to acquire two opposite macroscopic momentum. Which in turn  one can expect to forbid a spatial cat a priori.  The dimensions of \( w \) are \( J/m^2 \) and \( [kg]^{-1} \) for \( X_i \) and \( P_i \), respectively.

		\vspace{0.5cm}
		Symmetry by translation and rotation is respected by these cases, therefore usual conservation laws are expected.

		\section{The $O_i=J_i=L_i+s_i$ case}\label{sec:Jcase}
		
		I ended section \ref{sec:intro} with the question of whether the case $ O_i=L_i+s_i $ could be an option for satisfying a),b) and c). 
		Therefore, I now turn  to answer this question. 
		In \eqref{eq:cancellation} I have 
		\begin{eqnarray*}\label{key}
			G(O_1,\dots, O_N)= (L_k+s_k)^2 + 2 \underset{j:j\neq k}{\sum}(L_k+s_k)(L_j+s_j) - 2 \left(\sumi (L_i+s_i)(\psi)\right)(L_k+s_k)= I + II +III,
		\end{eqnarray*}
		where
		\begin{eqnarray}\label{eq:I-II-III}
			I =& L^2_k + 2 \underset{j:j\neq k}{\sum}L_kL_j - 2 \left(\sumi (L_i(\psi))\right)L_k,\\
			II =&  s^2_k + 2 \underset{j:j\neq k}{\sum}s_k s_j - 2 \left(\sumi (s_i(\psi))\right)s_k,\\
			III =&  2L_k s_k + 2 \underset{j:j\neq k}{\sum}(L_k s_j+ s_k L_j) - 2 \left(\sumi (s_i(\psi))\right)L_k -  2 \left(\sumi (L_i(\psi))\right)s_k.
		\end{eqnarray}
		We check if these terms will add or not extra terms in \eqref{eq:cinematic} and \eqref{eq:COMotion}. 
		
		\vspace{0.6cm}
		
		To evaluate the angular momentum $ O_i=L_i $, we  consider the two-dimensional case $ (X_k,Y_k) $ where $ L_k:= L_{z,k}= X_k P_{y,k}- Y_k P_{x,k}$. Moreover, the wavefunction will be $ \psi(x_1,\dots, x_N)\chi(s_1,\dots, s_N) $ for short $ \psi(x)\chi(s) $, where $ \chi $ is on a finite state space.
		We need to introduce some commutators: $ [X,L_z]= -i\hbar Y $ and $[X, L^2_z] = -i \hbar (YL_z +L_z Y) $. Using  these commutators for  term $ I $, for the $X-$component, the term \eqref{eq:cancellation} is:
		\begin{equation*}
			\mathcal{F}_k(\psi):=-i\hbar \bra \psi \dverti{ Y_k L_k + L_k Y_k}\psi \ket - 2i \hbar\underset{j:j\neq k}{\sum}\bra \psi \dverti{L_j Y_k}\psi \ket + 2i\hbar \mathcal{L}(\psi) \bra \psi \dverti{Y_k} \psi \ket,
		\end{equation*}
		giving in $ \frac{d \left\langle\psi | X_k| \psi\right\rangle}{dt} $ the extra term:
		\begin{equation*}\label{eq:nonzero}
			m\frac{d X_k(\psi)}{dt} = P_{x,k}(\psi) +wm\frac{\mathcal{F}_k(\psi)}{i\hbar},
		\end{equation*}
		and consequently for the COM ($ M=mN $)
		\begin{equation}\label{extraNewton}
			M\dot{X}(\psi)= P_{x}(\psi)+ wm \sumi\frac{\mathcal{F}_i(\psi)}{i\hbar}.
		\end{equation}
		For a general $\psi$ there is no reason to expect that the last term is zero, therefore I conclude that COM motion dynamics is violated. This can be verified with a simple example. Consider a single particle ($N=1$) in a 2D Gaussian state centered at the origin with average momentum $p_0$ along $x$:
		\begin{equation*}
			\psi(x,y) \propto e^{-\frac{x^2+y^2}{2\sigma^2}} e^{i \frac{p_0}{\hbar} x}.
		\end{equation*}
		
		We have $\langle L_z \rangle = 0$ and $\langle Y \rangle = 0$, so the last term of $\mathcal{F}_k$ vanishes. However, evaluating the anti-commutator $\{Y, L_z\} = 2XYP_y - 2Y^2P_x - i\hbar X$ yields a non-zero expectation given  by the term $-2\langle Y^2 \rangle \langle P_x \rangle = -\sigma^2 p_0$.
		Consequently, $\mathcal{F}(\psi) = -i\hbar(-\sigma^2 p_0) = i\hbar \sigma^2 p_0 \neq 0$. This implies that the velocity of the center of mass would anomalously depend on the spatial width $\sigma$ of the wavepacket, violating standard kinematics. Even if we have already seen that $ I $ should be rejected, we complete the check for $ II $ and $ III $, surprisingly the latter  cancels out. $ II $ cancels because  it involves commutations between space operators and spin operators. While for III, equation \eqref{eq:cancellation} becomes 
		\begin{equation*}
			-2i\hbar \bra \psi,\chi\dverti{\sumi s_i Y_k - \left(\sumi s_i(\chi)\right) Y_k	}\psi,\chi\ket=0.
		\end{equation*}
		
		An additional argument  considers the   term \eqref{eq:genwfe} under translation by a vector $a$.  With $O_i=X_i $ and $O_i=P_i$, it has the form of variance respectively for the position of COM and its momentum, which is invariant  under translation, while considering $O_i =L_i = X_i \times P_i $, and $L=\sum_i( X_i \times P_i)$ the total angular momentum, we have
		\begin{eqnarray*}
			D(\psi):=  \bra\psi \dverti{  \left[ \frac{1}{N}\sum_i((X_i + a) \times P_i) \right]^2  } \, \psi \ket - \left[\bra\psi \dverti{ \frac{1}{N}\sum_i ((X_i + a) \times P_i  ) } \, \psi \ket \right] ^2= \\
			\frac{1}{N^2}\left\{\left\{[L]^2(\psi)  - \left( L(\psi) \right)^2\right\} +   \left\{[L\cdot(a \times P)](\psi)  + [(a \times P)\cdot L](\psi) -  2L(\psi )(a \times P)(\psi)   \right\} + \left\{ [a \times P]^2 (\psi) - (a \times P(\psi))^2 \right\} \right\}.
		\end{eqnarray*} 
		In general  the second and third term are not zero.  The lack of invariance by translation of the term $wN^2 D(\psi)$ with $O_i = X_i \times P_i$ confirms our conclusion.
		
		\vspace{0.5cm}
		
		Having concluded that the case $O_i = L_i + s_i$ does not satisfy the constraints a), b), and c), and given that the acceptable cases ($O_i = P_i$ and $O_i = X_i$) of WFE are proposed  to eliminate spatial dispersion through nonlinear modifications to Schrödinger’s equation and not for sums of discrete “spins” as  examined in \cite{DeCarlo-Wick}, an interpretation of \eqref{eq:mag} is  that  WFE should apply to a magnetometer, not  to the system of spins or magnetic atoms. . Another possibility might be considering a model of magnet with both spatial and spin coordinates, discretizing  wavefunction states on a lattice with space `$a$' in presence both spin interactions and WFE,  to explore the indirect role of a spatial WFE on the spin coordinates. But this is  a future  research program.

		\section{Ideal test: a toy model and experimental estimates}\label{sec:HWMeasurment}

		In section \ref{sec:intro}, we mentioned that it should also be demonstrated that, along with $\alpha$ and $\beta$ in the initial state \eqref{eq:singlespin}, the apparatus system will display chaos, either moving to the right or left.  Here,  the idea is illustrated in the next subsection \label{app:toy} with a discrete  system (introduced in \cite{MP3}), where a spin variable undergoes diffusion. The experiment closest to the toy model we present is \cite{Abdietl}. Here they propose a a mechanical resonator consisting “lithium-decorated monolayer graphene sheet” of diameter one micrometer suspended in a “controllable, electrostatic double-well potential”. The metallic lithium render the wafer electrically conductive. The system is driven dissipatively toward its ground state via optomechanical sideband cooling. Observation is by way of magnetic coupling to a “superconducting qubit”.
		
		From the website \cite{Q-Tube}, it seems that the ERC Q-Tube aims to create cat states using double-well potentials.

		\subsection{A  spin toy model and chaos (property d))}\label{app:toy}

		The    microscopic system is represented by one qubit with spin $ J=1/2$, entangled with a system of $N-1$ spins forming the ``apparatus'', with readout given by the total spin $S=\overset{N}{\underset{i=2}{\sum}} s_i$. The state of the system will be 
		\begin{equation*}
			\psi = \underset{s:=\{s_1,\dots,s_N\}}{\sum}\psi(s_1,s_2,\dots, s_N)\verti s_1,s_2,\dots, s_N\rangle.
		\end{equation*}
		The linear part of the dynamics is:
		\begin{eqnarray}
			&H_{QM}   = -\frac{1}{m} \Delta  + V(S), \label{eq:toymod} \text{ where}\\
			\Delta \psi & = \sum_{i=1}^{N}\left\{\psi(s_1,\dots,s_{i}+1,\dots,s_N)+\psi(s_1,\dots,s_{i}-1,\dots,s_N)  - 2 \psi (s_1,\dots,s_i,\dots,s_N)\right\}
		\end{eqnarray}
		with reflecting boundary conditions, i.e. $\psi(s_1,\dots, s_i +1, \dots,s_N)= \psi(s_1, \dots, s_i,\dots, s_N)$ if $s_i +1 $ is greater than $1/2$ and  $\psi(s_1,\dots, s_i -1, \dots,s_N)= \psi(s_1, \dots, s_i,\dots, s_N)$ if $s_i -1 $ is smaller than $ - 1/2$. The external potential in fig. \ref{fig:Dwell}, is a double well
		\begin{eqnarray*}
			V(x)= {const}(x^2-R^2)^2,\\
			R= J(N-1),\,\,	const =  \Delta V/R^4.
		\end{eqnarray*}
		And additive constant can be added to $V(S)$. The system close to the ground state can form a cat, i.e. an object where the system is at the same time in the well on the right and in the one on the left. This means states like
		\begin{equation}\label{eq:spin cat}
			\psi = \frac{1}{\sqrt{2}} \left| S = R\ket  +\frac{1}{\sqrt{2}} \left| S=-R \ket,
		\end{equation}
		while \eqref{eq:genwfe} is written with $O_i =s_i$. Although the model is written for spins, the kinetic term and the nonlinearity should be intended as spatial, that is one spreads the spins into the wells and the  other bounds the separation between the two wells. Of course this is just an illustrative toy model, but it allows simulations since it is  finite dimensional.
		Finally, the spin $s_1$ is coupled with the apparatus with a term 
		\begin{equation}
			\alpha s_1 S.
		\end{equation}
		Initially, the system has to be confined (in dispersion)  to a narrow band centered at the local of the central unstable-equilibrium point ("hill") in the potential. So for the initial state one considers
		\begin{equation}\label{eq:initialst}
			\psi = Z_N \left[\alpha \left| 1/2 \ket + \beta \left|-1/2\ket\right] \times \left[  \underset{s_i=\pm 1/2: \dverti{S}\leq center }{\sum}  \left|s_2\ket \left |s_3\ket \dots \left|s_N\ket   \right],\,\, \alpha^2 +\beta^2=1,
		\end{equation}
		where $Z_N$ is a normalization constant and  $center$ is  a parameter to center the system. 
		The idea is that if the initial state has a small  asymmetry in the spin to ``measure'', that is $\alpha = \gamma +\epsilon$ and $\beta= \gamma - \epsilon$ ($<s_1>= 2\epsilon \gamma$), then, if the non-linearity is  strong enough,  the system will be pushed into one of the two wells, giving a definite result, i.e.,  ``spin-up'' or ``spin-down''. Thus, the randomness results from the  macroscopic  amplification of a small instability. (On a long time scale, to have the system setting to rest, we should introduce  friction.)   If $w=0$, the model will disperse. I formalize this concept by applying the Theorem at page 10 of \cite{MP3}, where the model was simulated.  First, I need to  write the symplectic system of ODEs along with its linearized system on the tangent plane.  The Hamiltonian takes the form:
		$
		E = \langle \psi | H_{QM} | \psi \rangle + w \left\{ \langle \psi | S^2 | \psi \rangle - \langle \psi | S | \psi \rangle^2 \right\}.$
		Next, the wavefunction is decomposed into real and imaginary parts as: $\psi_k = Q_k + i P_k$, for the initial state \eqref{eq:initialst} $P_k=0$. We  then have the Hamiltonian system
		\begin{equation}
			\begin{aligned}
				\frac{\partial Q_k}{\partial t} &=\left( \frac{1}{2}\right) \phantom{-} \frac{\partial}{\partial P_k} E(Q,P) \\
				\frac{\partial P_k}{\partial t} &= - \left( \frac{1}{2}\right) \frac{\partial}{\partial Q_k} E(Q,P)
			\end{aligned},
			\qquad 
			\begin{aligned}
				E(Q,P) = (1/2) \sum_{j,k} P_j H_{QM,j,k} P_k &+ (1/2) \sum_{j,k} Q_j H_{QM,j,k} Q_k 
				+ w \left\{ \sum_k (P_k^2 + Q_k^2) S_k^2 - \left[ \sum_k (P_k^2 + Q_k^2) S_k \right]^2 \right\},
			\end{aligned}
		\end{equation}
		where '$k$' here indexes a configuration and $H_{QM} $ a real, symmetrical matrix. Let $\#$ be the total number of spin configurations. The associated linearized  dynamical system,  which approximates the motion of the  original for a small time interval, is  defined by:
		\begin{equation}\label{eq:linsys}
			\begin{aligned}
				\frac{\partial \xi_k}{\partial t} &= \left( \frac{1}{2}\right) \phantom{-} \sum_{j} \left\{ \frac{\partial^2 E}{\partial P_k \partial Q_j} \xi_j + \frac{\partial^2 E}{\partial P_k \partial P_j} \eta_j \right\} \\
				\frac{\partial \eta_k}{\partial t} &= -\left( \frac{1}{2}\right) \sum_{j} \left\{ \frac{\partial^2 E}{\partial Q_k \partial Q_j} \xi_j + \frac{\partial^2 E}{\partial Q_k \partial P_j} \eta_j \right\}
			\end{aligned}
		\end{equation}
		Here $\xi$ and $\eta$ are real $\#$-vectors, in matrix form \eqref{eq:linsys} can be written:
		\begin{equation}
			\frac{d}{dt} \begin{pmatrix} \xi \\ \eta \end{pmatrix} = M(t) \begin{pmatrix} \xi \\ \eta \end{pmatrix}
			\qquad \text{with} \quad
			M = \begin{pmatrix} A & B \\ C & D \end{pmatrix}
			\quad \text{where} \quad
			\begin{aligned}
				A &= -u \otimes v \\
				D &= -A^T \\
				B &= \Lambda - u \otimes u \\
				C &= -\Lambda + v \otimes v
			\end{aligned}.
		\end{equation}
		Here $u$ and $v$ are $\#$-vectors and $\Lambda$  a matrix $\#\times \#$ . Specializing to the model \eqref{eq:toymod}:
		\begin{equation}\label{eq:matrix}
			H = H_{QM} + \text{diag}(f), \quad f_i = w S_i \left\{ S_i - 2 \sum_{k} S_k (Q_k^2 + P_k^2) \right\}, \quad
			\begin{aligned}
				A_{i,j} &= - 4 w S_i S_j P_i Q_j \\
				B_{i,j} &= H_{QM,i,j} + f_i \delta_{i,j} - 4 w S_i S_j P_i P_j \\
				C_{i,j} &= - H_{QM,i,j} - f_i \delta_{i,j} + 4 w S_i S_j Q_i Q_j \\
				D_{i,j} &= - A_{j,i}
			\end{aligned}\,\,,
		\end{equation}
		which implies
		\begin{equation}
			\Lambda = H_{QM} + \text{diag}(f), \quad u_i = 2\sqrt{w} S_i P_i, \quad v_i = 2\sqrt{w} S_i Q_i.
		\end{equation}
		Expanding and contracting directions, namely directions $(\xi,\eta)$ on the tangent plan where a small difference in the initial condition is amplified in time, appear if the matrix $M$ has positive and negative real eigenvalues. This happens if the following theorem, introduced in \cite{MP3}, is verified:
		\begin{theorem}
			Let $M$ be an even-dimensional ($2\# \times 2\#$) matrix and
			\begin{equation}\label{eq:DCI}
				\det M < 0,
			\end{equation}
			then $M$ has both positive and negative eigenvalues.
		\end{theorem}
		The proof is in Appendix \ref{app:chaos_cond}. Now I apply this result to show that exist a $w^*>0$ such that ``Determinant Criterion of Instability''(DCI) \eqref{eq:DCI} $\det M <0$ is satisfied. It is important to observe that  because the determinant condition takes the form of an inequality, every such
		wavefunction will exist in a neighborhood also satisfying the condition. Therefore,
		a system started at such a point will enjoy that condition at least for a short
		time.
		
		\begin{lemma}
			For the case of the matrix  \eqref{eq:matrix} defined above,  the condition \eqref{eq:DCI}  is equivalent to 
			\begin{equation}\label{eq:DCI2}
				v^t \Lambda^{-1} v > 1.
			\end{equation}
		\end{lemma}
		Again, we refer the proof to Appendix \ref{app:chaos_cond}. Considering the apparatus and plugging-in the flat superposition \eqref{eq:initialst} into the expression 
		\begin{equation}\label{eq:x}
			r= 4w Q^T \text{diag}(S)\Lambda^{-1} \text{diag}(S)Q = 4w Q^T \text{diag}(S)\frac{1}{H_{QM} + \text{diag}(f)} \text{diag}(S)Q >1.
		\end{equation}
		If $w=0$, then $ r=0$, i.e., for the linear theory there are not instabilities.  $H_{QM}$ is a real symmetric matrix, define it such that $H_{QM}\geq 0$ or such that $\Lambda$ is invertible in the neighborhood of interest, choosing odd the number of spins composing the apparatus  and $\langle S\rangle <1/4$ to avoid invertibility troubles, for the limit $w\rightarrow\infty$ it is
		\begin{equation}\label{eq:wlim}
			\underset{w \rightarrow \infty}{\lim} r= 4\sum_k Q_k^2 \left(\frac{S_k}{S_k- 2\langle S \rangle }\right)   >1.
		\end{equation} 
		Other choices are possible, e.g., $H_{QM} >0$, any parity of the number of spins and  $\underset{k:S_k\neq 0}{\sum} Q_k^2 > 1/4 $.
		
		\begin{corollary}
			There exists $0<w^*< +\infty$ such that for $w>w^*$, the DCI condition is satisfied, defined by $w^*=\sup \left\{w: r(w)=1\right\}$.	
		\end{corollary} 
		This follows from the fact  $r(w)$ is a continuous function, then there exists $w^*>0$ such that $r(w)>1$ for $w>w^*$.
		For the choice of an initial centered state, or perturbations in its neighborhood, $\langle S \rangle \approx 0 $; therefore $\underset{w \rightarrow \infty}{\lim} r \approx 4\sum_k Q_k^2= 4>1$ and the DCI condition is satisfied.
		
		It would be interesting to know if the strength of $w$ decreases with an increasing  number of spins. In \cite{MP3}, there is  some numerical evidence of this, albeit limited  because of  computational resources. Given the increase in dimensionality in wavefunction models, it would be difficult to scale the number of spins in simulations (an application of parallel computing could be of help).

		Even if  a toy model, it is an high-dimensional nonseparable nonlinear Hamiltonian model, so the best way to study it is via numerical methods. Considering various randomizations of the initial conditions, the simulations of \cite{MP3} support the picture presented here, with the ``apparatus'' moving right or left when the nonlinearity becomes strong enough to satisfy DCI condition. In the linear case $w=0$, the simulations of  the system  evolved into a cat state. In \cite{ContinuumChaos},  it was   shown numerically that the Maximal Lyapunov exponent is positive in these type of models, which is considered the signature of deterministic chaos.
		
		The picture rises the fascinating hypothesis of connecting the threshold where cats disappear with that of deterministic chaos \cite{ContinuumChaos}; here, we are generalizing the DCI condition to continuum models. For the  \eqref{eq:toymod} and different randomizations of the initial condition, a comparison with the Born's rule was made in \cite{MP3}. In this regard, in \cite{DotonScreen}, a statistical comparison with Born's rule was made for a particle-detector model.
		
		The computational techniques of \cite{MP3} are described therein and  are developed from \cite{NumRec} and \cite{Tao}.   
		
		\begin{figure}[htbp]
			\centering
			\includegraphics[width=0.8\textwidth]{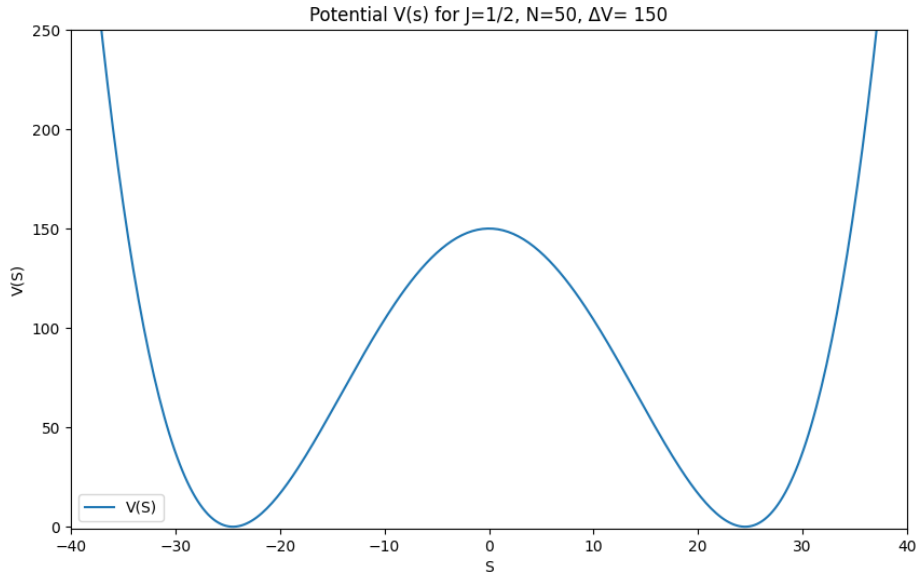}
			\caption{External potential vs apparatus spin}
			\label{fig:Dwell}
		\end{figure}
		
		\subsection{The energy barrier}

		The idea behind this ideal experiment  is that WFE is an energy barrier to form cat states and to persist energy has to be supplied by an external potential. Therefore I presented the double-well toy model. Consider a drop of $ \Delta V $ in the potential energy between the origin and the minima.  The case $O_i=X_i$ is more intuitive. In this case one expects  a critical $N_c$, where cats stop to form, such that 
		\begin{equation}\label{eq:critest}
			\Delta V=N_c \Delta v \approx w N_c^2 R^2\,\,i.e.\,\, w= \frac{\Delta V}{N^2_c R^2}= \frac{\Delta v}{N_c R^2}
		\end{equation}
		We see  that boundary is not only a question of size but also of energy, namely being able to  supply some extra energy $ \Delta V'= \Delta V +\delta V>  \Delta V$, there will be some extra room for cats. This was observed in the simulations of \cite{MP4}, where $w$ is artificially tuned large for small $N$ in the toy model we described. 
		Of course scalability is an important issue for the realization of the experiment.
		
		I attempt to plug some reasonable orders of magnitude into \eqref{eq:critest}. Considering the temperature of $T=1$ nK, I estimate the height of the barrier using the suggestion of Ashkin et al. \cite{Ashkin1986} for optical traps, that is $\Delta v \sim 10-10^2 k_BT$, where $k_B$ is the Boltzmann constant. With $N \sim 10^9$ (typically a radius of $100$ nm) and $R\sim 10^{-7}$ m, I get $w \sim 10^{-25}-10^{-26}$ SI. Of course this is just an example; without a realization of the experiment it is quite difficult to set parameters. The experiment closest to the toy model I presented is \cite{Abdietl}. In their proposal, the mechanical resonator is a monolayer graphene sheet with mass $M \approx 5.7 \times 10^{-19}$ kg ($5.7 \times 10^{-16}$ g), which corresponds to roughly $N \approx 3 \times 10^7$ carbon atoms, moving in a symmetric
		double-well potential. The engineered electrostatic potential is $V(x) = -\frac{\nu}{2}x^2 + \frac{\beta}{4}x^4$, where the parameters $\nu \approx 1.3 \times 10^{-5}$ J/m$^2$ and $\beta \approx 3.7 \times 10^{15}$ J/m$^4$ create a barrier height of $\Delta V = \nu^2/4\beta \approx 1.1 \times 10^{-26}$ J and a well separation of $R \approx 2\sqrt{\nu/\beta} \approx 1.2 \times 10^{-10}$ m. Plugging these values into \eqref{eq:critest}, we obtain a rough estimate of $w \sim 10^{-21}$ J/m$^2$. Maybe we could be lucky with the project ``Q-Tube'' granted in 2025 \cite{Q-Tube}, where they seem to aim to try something as suggested here; their \emph{`objective is to achieve quantum delocalization that exceeds the size of the nanotube using the double-well potential approach''}.

		For the other admissible case found in Section \ref{sec:derivationcom}, it is possible to make specular considerations, reasoning that a spatial cat to form,  it has first to ``diffuse''. Consider the following initial wavepacket in the $p$ representation
		\begin{equation}\label{eq:singlespin}
			\alpha\psi_q(p -p_0)  \left| +\ket + \beta e^{i\gamma} \psi_q(p +p_0)\left| -\ket,
		\end{equation}
		representing a particle with spin in a superposition up-down and possessing momentum of opposite sign. Once  entangled with a system acting as an apparatus at $x $,  initially in the  state $\pN \phi_q(p_j)e^{ikx} $, which is meant  to move either to the right or left to produce an outcome, the term \eqref{eq:genwfe} with $ O_i=P_i $ prevents  the  system's evolution into a cat state $\frac{1}{\sqrt{2}}\psi_q(p -p'_0)    \pN \phi_q(p_j - P) \left| +\ket + \frac{1}{\sqrt{2}} e^{i\gamma'} \psi_q(p +p'_0)\pN \phi_q(p_j + P)\left| -\ket $.  With calculations similar to those in  Section \ref{sec:intro}   the initial state  has energy $E_{QM}(\psi)+E_{WFE}(\psi)$  of the order $N$, while the cat state scales as $ N^2 $.      In Section \ref{sec:comparison}  I will show that bounding the dispersion of the momentum reflects on the dispersion of the position. 
		
		In the case of P-WFE, the ``critical'' estimate \eqref{eq:critest} is less intuitive and would require a dynamical analysis, perhaps similar to the one performed for free expansion in Section \ref{sec:comparison}, but adding a double well  potential. Considering that the important parameters of the experiment are the spatial dimension $R$ and the barrier height $\Delta v$, a naive estimate is given by a momentum cat where $P = m R/T$, with $T$ being the time required for the system to diffuse into the wells, i.e.,
		\begin{equation}\label{eq:Pestimate}
			\Delta V \approx w N_c^2 P^2 = wm^2 N_c^2 \frac{R^2}{T^2} = w' N_c^2 R^2.
		\end{equation}
		Another way to see \eqref{eq:Pestimate} is considering $P=MR/T_N$ with  $T_N= N\tau $, i.e.  making a scaling assumption about the time to form a momentum cat.
		
		A more direct approach for testing P-WFE experimentally is suggested by \cite{HigbieStamper}, who describe a procedure to generate momentum cat states using a double-well potential in $P$. In their proposal, a Bose-Einstein condensate is subjected to a pair of counter-propagating laser beams that induce Raman transitions between internal atomic states.  They engineer a dispersion relation that is a double-well potential in momentum space. The atoms are driven into a superposition of two distinct momentum states, effectively creating a `momentum cat' where the wavepacket is split not in position, but in velocity space. Their motivation is that creating initial momentum cats is feasible with larger atom numbers than spatial cats. This is because, instead of relying on unstable attractive condensates, repulsive condensates can be used, which facilitates scaling the number of atoms. Furthermore, while creating a spatial double well requires ``exceptional spatial control'' over trapping potentials with extremely small dimensions, the parameters for the momentum-space scheme (momentum well spacing, barrier height) can be controlled robustly by tuning laser frequencies and angles. The system is confined in a harmonic potential which plays the role of the ``kinetic energy'' term; in the momentum representation, the (rescaled) effective Hamiltonian is:
		\begin{equation}\label{eq:PHam}
			H_P = \left[P^2 - \frac{1}{2}\sqrt{4k^2 P^2 + \Omega^2}\right] - \left[\frac{1 }{2\tilde{m}}\frac{\partial^2}{\partial P^2}\right].
		\end{equation}
		Here $P$ denotes the one-dimensional momentum variable, and $\tilde{m} $ is the corresponding dimensionless effective mass parameter. In this setup, the shape of the potential is determined by the external laser. The parameter $k$ represents the momentum transfer imparted by the Raman lasers ($\hbar k = \hbar|\mathbf{k}_1 - \mathbf{k}_2|$), which sets the separation distance between the two momentum wells. The parameter $\Omega$ is the two-photon Rabi frequency, proportional to the laser intensity, which controls the height of the barrier; a smaller $\Omega$ results in a higher barrier and more distinct wells.
		This potential forms a double well when $\Omega < k^2/2$. In this regime, there are two minima located at $\sqrt{4k^2 P^2 + \Omega^2} = k^2$, separated by a barrier of height $\Delta v = \frac{1}{4} \left( k - \frac{\Omega}{k} \right)^2$.  Now the equivalent of \eqref{eq:Pestimate}    is 
		\begin{equation}
			\Delta V  = N \frac{1}{4} \left( k - \frac{\Omega}{k} \right)^2   \approx w N_c^2 P^2 .
		\end{equation}
		Unfortunately, \cite{HigbieStamper} only describes how the experiment should be performed. In \cite{Sadler2004}, the authors state they were preparing the experiment, but due to lack of access to the full paper, I could not determine if they provided values for these parameters.

		\section{Comparison with Collapse models}\label{sec:comparison}
		
		I devote this Section to a comparison between the idea presented here, that is  supplementing the Schr\"odinger's evolution with a term that forbids macroscopic superpositions by energy conservation, and the well known Collapse Models, where a stochastic dynamics supplements Schr\"odinger's evolution to model an objective collapse. For this, I will try to identify  what they have in common and what they do not,  from a general point of view, focusing later on the differences in mathematical structure used to achieve their goals, and finally on their most representative phenomenology.  Since their introduction with the GRW model \cite{GRW86}, objective collapse theory has evolved and various branches  have emerged\cite{Collapse}.  For the present discussion, we consider as reference  the Continuous Spontaneous Localization (CSL) model introduced by Ghirardi, Pearle and Rimini \cite{GPR90}, which is the one around which most of the current models are built.
		
		\vspace{0.3cm}
		
		While in a Schr\"odingerist philosophy I treat the wavefunction as an element of reality, Collapse models are more ambiguous on their meaning; they can be intended in that way but more typically they are considered as phenomenological descriptions of an underlying physics \cite{Collapse}.  The  models I presented here and Collapse models share the idea that to restore the behavior of matter moving as point particles under forces, it is necessary to block the spreading in space of the wavefunction state.  The modeling issue is to supply Schr\"odinger's equation with new terms, agreeing  that superpositions do not disappear because of the limit $\hbar \rightarrow 0$ (on this, I observed that in the derivation of Section \ref{sec:derivationcom} that $\hbar$ cancels out) and decoherence. In connection with the latter, I observe that interference terms already contribute negligibly to the dispersion of large objects.	To show it, I consider a  test function $\phi_r$ as in the introduction \ref{sec:intro}, where $r$ is larger than $R$, and compute the Dispersion. We define the test function:
		$$
		\phi_r^2(x) := \begin{cases} 
			\frac{1}{2r} & x \in [-r, r], \text{ with }  r=\frac{3}{2}R \\
			0 & \text{otherwise}
		\end{cases} %
		$$
		and  the Cat state $
		\psi = \frac{1}{\sqrt{2}} \prod_i \phi_r(x_i + R) + \frac{1}{\sqrt{2}} \prod_i \phi_r(x_i - R) = \frac{\psi_R}{\sqrt{2}} + \frac{\psi_{-R}}{\sqrt{2}}. 
		$
		Let $S$ be the single-particle overlap integral: $S = \int_{-\infty}^{\infty} \phi_r^*(x + R) \phi_r(x - R) \, dx= 1-\frac{R}{r}$. It follows that
		$$ \| \psi \|^2 = 1 + S^N, $$
		meaning the interference terms go to zero with $N$ since $S<1$. For the dispersion  is 
		\begin{equation}\label{eq:D-S}
			D_X =  \langle\left(\sum X_i\right)^2 \rangle = \frac{Nr^2}{3} + N^2 R^2 + N \frac{R^2}{36} S^{N-1}
		\end{equation}
		The interference term $S^{N-1}$ contributes to \eqref{eq:D-S} with a term becoming exponentially small, even without including environments.  This also explains why looking for the ``Infamous Boundary'' by studying interferometric experiments is probably not the best approach.
		
		The mathematical structure to block these types of states is quite different between WFE and Collapse Models. To better appreciate this, we describe the CSL dynamical equation for the wave function defined in \cite{GPR90}. This  is written as
		\begin{equation*}
			d|\psi_t\rangle = \left[ -\frac{i}{\hbar}H dt + \sqrt{\gamma} \int dx (N(x) - \langle N(x) \rangle_t) dW_t(x) - \frac{\gamma}{2} \int dx (N(x) - \langle N(x) \rangle_t)^2 dt \right] |\psi_t\rangle,
		\end{equation*}
		where $H$ is the Hamiltonian describing the standard quantum mechanical dynamics, and $\gamma = \lambda (4\pi/\alpha)^{3/2}$ is a coupling parameter measuring the strength of the collapse, where $\lambda$ is the collapse rate and $\alpha = 1/r_c^2$ is the localization parameter determined by the localization length $r_c$. The term $\langle N(x) \rangle_t$ is the expectation value of the locally averaged particle number density operator $N(x)$, defined as
		\begin{equation}\label{eq:Ndef}
			N(x) = \sum_s \int dy \, g(y-x) a^\dagger(y, s) a(y, s),
		\end{equation}
		where $a^\dagger(y, s)$ and $a(y, s)$ are the creation and annihilation operators for a particle with spin $s$ at point $y$. The smearing function is a Gaussian $g(x) = (\alpha/2\pi)^{3/2} e^{-\frac{\alpha}{2}x^2}$. Finally, the noise term is defined through a real Wiener process $W_t(x)$. The stochastic field continuously kicks the system, inducing a diffusion process that drives the wave function to localize exponentially fast onto the eigenstates of the particle density with a rate $\Lambda= N\lambda$ proportional to the number of particles, while the squared terms preserve the norm.

		WFE models present a Hamiltonian structure with an energy barrier, so the energy is exactly conserved and dynamics is time-reversible. In contrast, collapse models add stochastic diffusive terms and therefore lose this traditional characteristic of physics; energy is continuously pumped into the system and the dynamics are not reversible. In the framework of WFE models, traditional symmetries/variational analysis can be done.This will reflect in a different possible phenomenology to observe, as discussed in subsection \ref{subs:phenomena}. Both dynamics introduce nonlinear terms applied over the whole system that push the system around a definite macroscopic position.
		 
		Collapse models have an intrinsic stochastic nature, while in WFE models ``stochasticity'' is expected to come from deterministic chaos, in the sense of sensitivity on initial conditions, see section \ref{sec:HWMeasurment}. One interpretation of the physical origin of Collapse Models is that they are  gravity-induced; indeed, the  Diósi-Penrose model can be formulated as a CSL model\cite{Collapse}.  This seems to postulate a missing  quantum gravity, I find it difficult to think that a force as weak as gravity  could be the dominant mechanism at ``some'' mesoscale.

		\subsection{Phenomenological Comparison: Collapse Models vs WFEs}\label{subs:phenomena}
		
		The main natural prediction comes from the observations of Q. Fu \cite{Fu97}. He observed that CSL models inevitably induce an increase in energy, with the direct consequence that charged particles should radiate spontaneously. Fu derived the rate of this spontaneous radiation emission for free electrons as:
		\begin{equation}
			\frac{d\Gamma(k)}{dk} = \frac{e^2 \lambda}{4\pi^2 {r_c}^2 m^2 k},
		\end{equation}
		where $k$ is the emitted photon energy, $m$ is the electron mass, and $\lambda$ and $r_c$ are the characteristic CSL parameters. Consequently, X-ray emission spectra serve as an experimental test; these experiments have been used to set bounds on the collapse parameters \cite{Carlesso2022}. In contrast, WFE models do not predict any spontaneous radiation; the main natural prediction for WFE models is the existence of an energy barrier in experiments creating spatial/momentum cat states. As discussed in Section \ref{sec:HWMeasurment}, we should find a critical size $N_c$ such that $ \Delta v \approx w N_c R^2$ where the presence or absence of cat states can be tuned, by tuning the height of the barrier.
		
		\subsubsection{Free Macroscopic Body and Cloud of Atoms}\label{sec:clouds}
		
		The other phenomenon I consider the most characteristic effect of Collapse models is an increase in temperature with consequently induced Brownian motion on the dynamics of any system \cite{Bilardello2016}. Therefore, the interest lies in studying the spatial variances and covariance of the system in a cloud of atoms. To see this, and since it is useful for our comparison, I need to report them for standard Quantum Mechanics following the reference scheme of \cite{GPR90} for a \emph{free} macroscopic body, that is, with no external potentials. Namely, following the standard notation of Quantum Mechanics with brackets (i.e., $O(\psi)=\langle O\rangle $), we look at
		\begin{align}
			D_X &:= \langle X^2 \rangle - \langle X\rangle^2, \\
			D_P &:= \langle P^2\rangle - \langle P \rangle^2, \\
			\text{cov}_{XP} &:= \frac{1}{2}\langle XP+PX \rangle - \langle X \rangle \langle P \rangle,	
		\end{align}
		for the COM of a free macroscopic body of mass $M= Nm$ with energy $E(\psi)=\left\langle \sum \frac{P_i^2}{2m} \right\rangle$. We can do this because for $X_{CM} = \frac{1}{N} \sum X_i$ and $P_{CM} = \sum P_i$ we have $[X_{CM}, P_{CM}] = i\hbar$ and $E(\psi) = \left\langle \frac{P_{CM}^2}{2M} \right\rangle + \frac{1}{2m}\left\langle \sum (P_i - P_{CM}/N)^2 \right\rangle = \langle\psi\left|H_{CM}\right|\psi\rangle + \langle\psi\left|H_{rel}\right|\psi\rangle$; since $[P_{CM}, H_{rel}] = 0$ and $[X_{CM}, H_{rel}] = 0$, the relative motion decouples from the COM motion. This happens also if we add interatomic potentials $U(x_i-x_j)$. In the following discussion, I will leave it to the reader to distinguish when the aspect of having $N$ degrees of freedom is relevant, i.e., whether $X$ refers to just one degree of freedom of mass $M$ or to the $N$ degrees of freedom $\frac{1}{N} \sum X_i$ (with $M=mN$). It is $E_{X-WFE}(\psi)= w N^2 D_X(\psi)$ and  $E_{P-WFE}(\psi)= w  D_P(\psi)$.   We have for Schr\"odinger's equation the following time derivatives: $\frac{d D_P}{dt} = 0$, $\frac{d}{dt}\text{cov}_{XP} = \frac{D_P}{M}$ and $\frac{d^2 D_X}{dt^2} = \frac{2}{M} \frac{d}{dt} \langle \text{cov}_{XP} \rangle = \frac{2}{(M)^2} D_P$, which means
		\begin{align}
			D_P(t) &= D_P(0) \tag{32.A} \label{eq:MQvar.A}, \\
			\operatorname{cov}_{XP}(t) &= \operatorname{cov}_{XP}(0) + \frac{D_P(0)}{M} t \tag{32.B} \label{eq:MQvar.B}, \\
			D_X(t) &= D_X(0) + \frac{2 \operatorname{cov}_{XP}(0)}{M} t + \frac{D_P(0)}{M^2} \, t^2 \tag{32.C} \label{eq:MQvar.C}.
		\end{align}
		We see clearly that the dispersion of the initial momentum will affect the expansion. In particular, these relations are of interest for the free expansion of clouds of cold atoms \cite{Carlesso2022}, which are considered a possible test for deviations from Schr\"odinger's equations. Taking a small $R$, if the cloud starts to expand from a state of two counter-propagating clouds (Momentum Cat):
		\begin{equation}
			\psi = \mathcal{N} \left( \underbrace{\prod_{i=1}^N \phi_r(x_i - R) e^{i k x_i}}_{\text{Cloud at } +R \text{ moving } \rightarrow P} + \underbrace{\prod_{i=1}^N \phi_r(x_i + R) e^{-i k x_i}}_{\text{Cloud at } -R \text{ moving } \leftarrow -P} \right), 
		\end{equation}
		where $k = P/\hbar$, the scalings are
		$ D_X\sim O(r^2/N + R^2)$ and $D_P \sim O(N^2 P^2)$, while for $\psi = \prod_{i} \mathcal{N}_i \left[ \phi_r(x_i - R)e^{i k x_i} + \phi_r(x_i + R)e^{-i k x_i} \right]$ they are $ D_X \sim O\left(\frac{r^2 + R^2}{N}\right)$ and $D_P \sim O(NP^2)$, and for a ``classical state'' $\psi = \prod_i \mathcal{N}_i \,\,\phi_r(x_i - R) e^{i k x_i}$ they are $D_X \sim O\left(\frac{V_{X_i}}{N}\right)$ and $D_P \sim O\left(N D_{P_i}\right)$, where $D_{X_{i}}=\langle X_i^2 \rangle - \langle X_i \rangle^2 $ and $D_{P_i} = \hbar^2 \int |\phi_r'(x)|^2 dx$.
		Moreover, neglecting the cross terms, also the initial covariances have different scalings: $\text{cov}_{XP}$ is of order $O(NRP)$ for cats, $O(RP)$ for product states, and $0$ for classical states. Assuming that we have a protocol like that of \cite{HigbieStamper}, described in section \ref{sec:HWMeasurment}, to generate initial momentum cats, when scaling $N$ larger, the quadratic ballistic term in \eqref{eq:MQvar.C} becomes important at some time scale for the cat state (scaling as $O(1)$ relative to the mass), while remaining $O(1/N)$ if $D_P$ stays of order $N$, making the expansion appear diffusive. This consideration is important for the later comparison with WFE-models.
		
		For CSL models, the modifications to the noise-averaged variances for the COM are given by \cite{GPR90}:
		\begin{equation}\label{eq:CSLvar}
			D^{\text{CSL}}_{X}(t) = D_X(t) + \gamma\delta\frac{\hbar^{2}}{6M^{2}}t^{3}, \quad 
			\text{cov}^{\text{CSL}}_{XP}(t) = \text{cov}_{XP}(t) + \gamma\delta\frac{\hbar^{2}}{4M}t^{2}, \quad 
			D^{\text{CSL}}_{P}(t) = D_{P}(t) + \gamma\delta\frac{\hbar^{2}}{2}t,
		\end{equation}
		under the assumption that the internal wavefunction is sufficiently localized, neglecting the coupling between the COM and internal degrees of freedom caused by the stochastic term. The geometric factor $\delta$ is defined by the integral of the squared gradient of the smoothed density profile (the convolution of the particle density with the Gaussian kernel $g$ in \eqref{eq:Ndef}). The extra terms reflect the propagation of CSL-induced momentum diffusion; as this effect is suppressed for macroscopic masses ($1/M^2$), experimental bounds are often derived from the expansion of cold atomic clouds, where, assuming the atoms move independently, analyses look at \eqref{eq:CSLvar} with the single-atom mass $m$ \cite{Bilardello2016}.
		
		While a similar analysis to that of \cite{Bilardello2016} could be performed for WFE models to set bounds on $w$—a task worthy of future investigation—my focus here is different. Given the previous considerations on \eqref{eq:MQvar.C} and the interest in clouds of momentum cat states (possibly created as described in \cite{HigbieStamper}), I suggest considering the COM dispersions of the entire cloud. In particular, Wavefunction Energy is not possessed by individual atoms but by the system as a whole, whereas in CSL each particle is associated with its own localization function $g$. Therefore, I present below how Eqs. \eqref{eq:MQvar.A}, \eqref{eq:MQvar.B} and \eqref{eq:MQvar.C} change, and discuss what to expect. We start from P-WFE, which is simpler.
		
		\vspace{0.5cm}
		
		For P-WFE, the equations have the same form but with an effective mass $\frac{1}{M^*} = \frac{1}{M} + 2w$:
		\begin{subequations} \label{eq:Pvar_system}
			\begin{align}
				D_P(t) &= D_P(0) \label{eq:Pvar.A} \\
				\operatorname{cov}_{XP}(t) &= \operatorname{cov}_{XP}(0) + \frac{D_P(0)}{M^*} t \label{eq:Pvar.B} \\
				D_X(t) &= D_X(0) + \frac{2 \operatorname{cov}_{XP}(0)}{M^*} t + \frac{D_P(0)}{(M^*)^2} \, t^2 \label{eq:Pvar.C}
			\end{align}
		\end{subequations}
		
		The effective mass $M^*$ is slightly lighter, but the main difference is that in an experiment where  initial states are created as  momentum cats, at some time scale the ballistic term will be negligible when P-WFE becomes dominant, inducing a transition from a diffusive-ballistic regime to a purely diffusive one. This occurs because the energy cost of P-WFE prevents the formation of states with $D_P \sim N^2$, suppressing the coefficient of the quadratic term. It is useful to consider introducing a friction $\frac{dP_i}{dt} = -\frac{\gamma}{m} P_i$, which for a cloud of $N$ particles becomes $\frac{d P}{dt} = -\Gamma P$ where $\Gamma = \gamma/m$. Solving the system $\frac{d D_P}{dt} = -2\Gamma D_P$, $\frac{d}{dt} \text{cov}_{XP} = \frac{D_P}{M^*} - \Gamma \text{cov}_{XP}$ and $\frac{d^2 D_X}{dt^2} + \Gamma \frac{d D_X}{dt} = \frac{2}{(M^*)^2} D_P(t)$, we obtain $D_P(t) = D_P(0)e^{-2\Gamma t}$, $\text{cov}_{XP} (t) = \text{cov}(0) e^{-\Gamma t} + \frac{D_P(0)}{M^*\Gamma} e^{-\Gamma t} (1 - e^{-\Gamma t})$, and $D_X(t) = D_X(0) + \frac{2}{\Gamma M^*} \text{cov}_{XP}(0) (1 - e^{-\Gamma t}) + \frac{D_P(0)}{(M^*)^2 \Gamma^2} (1 - e^{-\Gamma t})^2$. This yields the same scaling considerations as before, but with the cloud settling to an asymptotic limiting size.
		
		\vspace{0.5cm}
		
		In the case of X-WFE, the new equations are coupled and the variances will remain bounded, describing an ellipse in phase space. Since it is not possible to write the equations in the same way for a particle of mass $M$ and for a mass $M=Nm$ of $N$ particles, we write them directly with $N$ explicit (for one degree take $N=1$). Introducing the frequency $\Omega_N = \sqrt{\frac{8wN}{m}}$, we now have the system:
		\begin{subequations} \label{eq:Xvar_system}
			\renewcommand{\theequation}{\theparentequation.\Alph{equation}}
			\begin{align}
				\frac{d}{dt} D_X &= \frac{2}{M} \text{cov}_{XP} \label{eq:Xvar.A} \\
				\frac{d}{dt} D_P &= -4wN^2 \text{cov}_{XP} \label{eq:Xvar.B} \\
				\frac{d}{dt} \text{cov}_{XP} &= \frac{1}{M} D_P - 2wN^2 D_X \label{eq:Xvar.C}
			\end{align}
		\end{subequations}
		Defining the constant energy $\overline{E} = \frac{D_P}{2M} + wN^2 D_X$, we derive the harmonic equation for $D_X$:
		\begin{equation}
			\frac{d^2}{dt^2} D_X + \Omega_N^2 D_X = \frac{4\overline{E}}{M},\,\ \text{ where } \overline{E} = \frac{D_P(0)}{2M} + wN^2 D_X(0),
		\end{equation}
		therefore
		\begin{subequations} \label{eq:Xvar_system_sol}
			\renewcommand{\theequation}{\theparentequation.\Alph{equation}}
			\begin{align}
				D_X(t) &= \frac{\overline{E}}{2wN^2} + A \cos(\Omega_N t) + B \sin(\Omega_N t),\, \text{where }\,A = D_X(0) - \frac{\overline{E}}{2wN^2}\, \text{ and }B = \frac{2 \text{cov}(0)}{M \Omega_N} \label{eq:XDX.A}, \\
				D_P(t) &= D_P(0) + 2M w N^2 (D_X(0) - D_X(t)) \label{eq:XDP.A},\\
			\end{align}
		\end{subequations}
		with a similar oscillatory equation that can be written for $D_P(t)$. These are the equations of a forced harmonic oscillator, and they should be considered only for short times; for longer times it would be appropriate to introduce friction (as before) that will lead the cloud to rest at a final size. So we consider the expansion:
		\begin{equation}\label{eq:Dexp}
			D_X(t) \approx D_X(0) + \frac{2\text{Cov}(0)}{Nm} t + \left[ \frac{D_P(0)}{(Nm)^2} - \frac{2wN}{m} D_X(0) \right] t^2 + O(t^3).
		\end{equation}
		Looking at \eqref{eq:Xvar.C}, there is a competition between the cloud's tendency to expand and the self-trapping due to X-WFE. So the expansion will be slowed by the trapping term $\frac{2wN}{m}D_X(0)$. Assuming the system starts in an initial state such that $D_X(0) = r^2/N$, $D_P = N\hbar^2/4r^2$ and $\text{cov}_{XP} \approx 0$, when WFE is active, the phenomenon to observe in this model is that the cloud is ``self-trapped'' (i.e., the expansion is halted) when
		\begin{equation*}
			2wN^2 D_x(0) \approx \frac{D_P(0)}{Nm}
		\end{equation*}
		giving the critical condition
		\begin{equation}\label{eq:mcritical}
			2w N r^2 \approx \frac{\hbar^2}{4mr^2}.
		\end{equation}
		That is, the derivative \eqref{eq:Xvar.C} becomes negative when $m> \frac{\hbar^2}{8wN r^4}$. Note that the same condition can be obtained from the acceleration term in \eqref{eq:Dexp} becoming negative. 
		Starting from an initial state as described, the phenomenon to observe is that the effect of self-trapping induced by WFE should appear first for larger molecular masses, while still persisting for lighter ones; also, for a given $m$, tuning $N$ around the critical value allows the self-trapping effect to appear and disappear. The condition \eqref{eq:mcritical}, given a mass, can be used to estimate $w$. For example, considering the mass of $^{87}\text{Rb}$ (which are the typical atoms used in these cloud expansion experiments \cite{Kovachy2015}), $r =10^{-7}$ and $N= 10^{9}$; we get $w \approx 10^{-25} \, \text{J/m}^2$.
		
		While both $D_X$ and $D_P$ are bounded since $\overline{E} = \frac{D_P}{2M} + wN^2 D_X$ is constant, and adding friction one can expect that the dispersion $D_P$ does not become of order $N^2$ before the system goes to rest, there is no barrier forbidding initial momentum cat states. The same can be said for Collapse models.
		
		\subsection{Concluding considerations and future investigations}

		The philosophy underlying WFE models is that the spatial separation of a complex of atoms encounters an energy barrier that depends on the number of degrees of freedom. I have illustrated some of the consequences, regarding both experimental tests and novel points of view on Thermodynamics and Chaos. Although this work is still at an initial stage of development, I emphasize its importance. Typically, solutions to the Measurement Problem are restricted to that specific issue and typically do not point to consequences for the rest of Physics; however, if a theory is fundamentally important, it should have implications beyond the mere scope of offering an explanation for the Measurement Problem. Given that the WFE based on angular momentum was rejected in Section \ref{sec:Jcase}, I mentioned in this context the problem of developing Gibbs wavefunction ensembles on a lattice for models with spatial and spin coordinates. Another open question concerns the threshold of chaos given by the DCI condition in Section \ref{sec:HWMeasurment}. In \cite{ContinuumChaos}, it was proved that these types of toy models intermittently enter and exit ``islands of instability''; furthermore, the author provided a generalization of the DCI conditions applicable to continuous spatial models. In this direction, we are currently investigating the implications for models applicable to gases and fluids \cite{ContinuumChaosIslands}. Simulations of the toy model in Section \ref{sec:HWMeasurment} presented in \cite{MP4} showed a ``hockey stick'' graph around the critical estimate \eqref{eq:critest}, indicating that the dispersion is an increasing function of the barrier height, with the explanation that some extra energy is supplied. In contrast, for collapse models, the dispersion is expected to behave like a Heaviside step function, depending only on the size. Therefore, an interesting goal would be to repeat the study of Subsection \ref{sec:clouds} by adding an external potential to the atomic clouds. To complete the comparison with Collapse models, one should also study the behavior of the density matrix in the large $N$ limit—since in Collapse models \cite{GPR90} the density matrix is diagonalized in the position basis—and repeat the study of \cite{Bilardello2016} regarding single-particle dispersion to bound $w$.
		
		In the introduction, I listed principles a), b), c), and d) that must be satisfied. While the case $O_i=L_i+S_i$ was found to fail condition c) in Section \ref{sec:Jcase}, both X-WFE and P-WFE satisfy these principles. In Section \ref{sec:HWMeasurment}, we discussed condition d) for ``spatial-spin'' toy models; however, the conditions presented there are not applicable to continuous spatial models. In \cite{ContinuumChaos}, the instability condition \eqref{eq:DCI} was generalized to a form applicable to continuous models and shown to be verified for particles in a harmonic potential with X-WFE; the same has not yet been investigated for P-WFE.
		
		An important factor in determining which type of WFE to choose is compatibility with relativity. P-WFE appears to be the preferred choice, as it is defined with respect to the center of mass momentum (see Section 4 in \cite{MP1}). However, X-WFE might also be extended, at least to Special Relativity, by recurring to space-time wavefunctions \cite{Dirac1932,Lienert2017,compatibility}, where each particle possesses a space-time coordinate. On the other hand, in General Relativity, global coordinate systems may not exist; hence, the expression $\sum x_i$ may lose meaning, making the concept of a Center of Mass either subtle or impossible to define. We note that \cite{compatibility} addressed the inclusion of P-WFE into Einstein's equations. 
		
		A common objection to modifications of the Schr"odinger equation are no-go theorems \cite{Bassi2003} concerning the violation of relativistic causality. It is not my goal to address this here, other than to observe that such theorems often rest on circular arguments (see Section 7 of \cite{compatibility}). If I have a personal worry is  explaining discrete phenomena like absorption spectra, this has been addressed by A.O. Barut's and collaborators \cite{Barut1991,Barut1992} work in wavefunction physics. Using the non-linear coupling of the electron to its own electromagnetic field, they say to obtain
		(instead of the discrete spectrum) a spectral concentrations or resonances at certain energy values.

		\appendix
		
		\section{Details on COM motions}\label{app:Newton}

		Considering  the functional $F(\psi,\psi^*)=  \left\langle\psi | X_k| \psi\right\rangle $ in   \eqref{eq:hamevo} and $E(\psi, \psi^*)$ given by the Hamiltonian in  \eqref{eq:ScheqManyB},   using integration by parts,  only the kinetic terms with $i=k$ remain:
		\begin{equation}\label{eq:cinematic}
			\frac{d \left\langle\psi | X_k| \psi\right\rangle}{dt}= \frac{\hbar}{2m}\left\{ -i \left\langle  \Delta_k \psi   \left|X_k \right|\psi  \right\rangle + i \left\langle\psi  \left| X_k\right| \Delta_k\psi \right\rangle  \right\}= i\frac{\hbar}{m}\left\langle \frac{\partial\psi}{\partial x_k} \middle| \psi\ket = \frac{1}{m}\bra \psi \dverti{P_k}\psi \ket .
		\end{equation}
		Summing over all $ k $  and averaging over $ N $ it gives
		\begin{equation}\label{eq: PofCOM}
			M\dot{X} = P, \text{ where } M=Nm \text{ and } P:=\left\langle\psi\left|\overset{N}{\underset{i=1}{\sum}}\,P_i\right|\psi\right\rangle .
		\end{equation}
		It is important to note  that this is not enough to call \ref{eq: PofCOM} classical: given an object at an initial position $ X(0) $, it should not be allowed to acquire two opposite momentum. .
		Now I compute $m \frac{d^2 \left\langle\psi | X_k| \psi\right\rangle}{dt^2}$ from \eqref{eq:hamevo} where $F(\psi, \psi^*)= i\frac{\hbar}{m}\left\langle \frac{\partial\psi}{\partial x_k} \left|  \right. \psi\right\rangle$. From $\partial_k^\dagger= -\partial_k$ and $\frac{\partial F}{\partial \psi^*}= \frac{i\hbar}{m}(-\partial_k \psi)$, given a state $\phi$, one has$
		\left\langle \phi \middle| \frac{\partial F(\psi, \psi^*)}{\partial \psi^* }\right\rangle = \frac{i\hbar}{m} \left\langle \partial_k\phi \left| \right. \psi\right\rangle$,
		from which it follows
		\begin{equation*}
			\frac{d^2 \left\langle\psi | X_k| \psi\right\rangle}{dt^2} = \frac{1}{m}\left\langle \partial_k \psi | H\psi\right\rangle +\frac{1}{m}\bra  H\psi \left|\right. \partial_k\psi \ket.
		\end{equation*}
		Observing that $\bra \pa_k \psi \verti W \psi \ket + \bra W\psi \verti \pa_k \psi \ket = - \bra \psi \verti \pa_k W |\psi\ket $, where $W=V+U$,  one arrives at	$	m\frac{d^2 \left\langle\psi | X_k|\psi\right\rangle}{dt^2} = -\bra \psi|\pa_k W |\psi\ket.$
		Summing over all the particles and diving by $N$:
		\begin{equation}\label{eq:COMotion}
			M \ddot{X} = -\sumk \bra \psi | \pa_k W |\psi\ket,\,\, M=mN.
		\end{equation}
		I consider $U= \overset{N}{\underset{i\neq j}{\sum}}u(x_i-x_j)+ \tilde{u}\left(X,y\right)$ and $V=\underset{i=1}{\overset{N}{\sum}} v(x_i) $, where $\tilde{u}$ describe the interaction between the COM and the entagled particle.
		The term for $ \overset{N}{\underset{i\neq j}{\sum}}u(x_i-x_j)$ will disappear, e.g. consider typical interactions like Lennard-Jones or harmonic potential. (The intuitive reason is that the term describes internal bulk forces, so they will not move the COM.  This can be checked for solvable cases like some chains of harmonic oscillators.)  In this way we arrive at 
		\begin{equation}\label{eq:finalNew}
			M\ddot{X} = - \sumk \bra \psi | \pa_k v (x_k) |\psi\ket - \sumk \bra \psi | \pa_k \tilde{u} (X,y) |\psi\ket  .
		\end{equation}

		\section{Proof of spontaneous magnetization in wavefunction ensembles}\label{app:proof}
		
		\def\ooN{{\frac{1}{N}}}
		\def\sdel{{\sqrt{\delta}}}
		\def\mhalf{{-\frac{1}{2}}}
		\def\Amax{{A_{\hbox{max}}}}
		\def\Amin{{A_{\hbox{min}}}}
		\def\sumnN{{\sum_{n=0}^N}}

		Because of the indistinguishability condition we studied mean field models for the quadratic part of Quantum Mechanics $E_N(\psi)$, as  we could not devise a  way to construct lattice spin models for the case of  nearest neighbors. 
		To remove the notion of distinguishable particles  we considered the Curie-Weiss energy:
		\be\label{eq:CWham}
		E_{CW}(S) \= -\frac{1}{N}\, \left(\,\sum_{i=1}^N\,S_i\,\right)^2,
		\ee
		where we replaced spin configurations by \wfs\ with exchange symmetry; in our ensemble we considered only symmetric wavefunction, while one should consider also antisymmetric ones. Also we considered $\pm1$ as spin values, while other ones could be considered.   We adopted the simplest model to learn how to construct some mathematical tools for  these wave-mechanical models and  reasoning that the lowest possible dimensionality of the Hilbert space
		and mean field interaction define the first case to study  whether phase transitions appear or not. 
		Note that with two levels the CW energy depends only on the number, call it `$n$', of ``down" spins.
		Thus we was led to introduce \wfs\ that depend only on `$n$'.

		The ensemble was rewritten as 
		\bar\label{eq:SCWensemble}
		\no [\,F(\phi)\,]_{\beta}= \int_{\{||\phi||^2 = 1\}}d\phi\,\exp\{ -\beta\,E_{CW}(\phi) \}\,
		\frac{F(\phi)}{Z_N},\,\,\, Z_N = \frac{1}{\Sigma_N}\int_{\{||\phi||^2 = 1\}}\,d\phi\,\exp\{ -\beta\,E_{CW}(\phi) \},
		\ear
		where  the magnetic energy associated to $\phi$ becomes:
		\be
		E_{CW}(\phi) \= - \frac{1}{N}\, \sum_{n=0}^N\,|\phi_n|^2\,\left(\,N-2n\,\right)^2
		\ee
		with $ \phi_n $  the component of the wavefunction with $ n  $ "down" spins. At this point, we had observed that
		\be
		E_{CW}(\phi) \= -N\,\left\{\,m^2(\phi) + D(\phi)\,\right\},\label{qcal}
		\ee
		where
		{\small \bar
			\no m(\phi)= \frac{1}{N}\,\sum\,|\phi_n|^2\,\left(\,N-2n\,\right);\,\,\,\, D(\phi)= \frac{1}{N^2}\,\left\{\,\sum\,|\phi|^2\,\left(\,N-2n\,\right)^2 \- 
			\left[\,\sum\,|\phi_n|^2\,\left(\,N-2n\,\right)\,\right]^2 \,\right\}.
			\ear}
		The term $D(\phi)$ is exactly \eqref{eq:genwfe}  with $O_i =S_i$.	 The next ingredient to have a magnetization in the thermodynamics limit, namely suppressing cat states, was introduced as follows.  we defined:
		\be
		f \= N\,\beta\,\left\{\, 1 - m^2(\phi) - D(\phi) + N\,w\,D(\phi)\,\right\}.
		\ee
		Here we have added a term ($N\beta$) to make $f\geq 0$ and incorporated the non-quadratic term $wE_{WFE}(\phi) = wN^2 D(\phi)$ in the energy, where $w$ is a very small constant.
		The intuition for the latter choice comes from the observation that, lacking that term,
		$f$ can be small if  either $m^2$ is large or $D$ is large; incorporating
		the dispersion term with large enough $wN$, the last possibility should be suppressed.
		At the end the model was defined by:
		\bar
		\no [\,m^2\,]_{\beta}= \int_{||\phi||^2 = 1}\,d\phi\,\exp\{ -f(\phi) \}\,m^2(\phi)/Z_N;\,\,\,\,
		Z_N= \frac{1}{\Sigma_N}\int_{||\phi||^2 = 1}\,d\phi\,\exp\{ -f(\phi) \}.
		\ear
		Intuition suggested investigating cases where $w$ is at least $1/N$; 
		hence we defined
		\be
		\omega \= N\,w;
		\ee
		and we assumed that $\omega$ is a constant. This 
		does not indicate a belief that $w$ actually scales with $N$; if such no-quadratic terms have a physical correspondence, then $w$ is a constant  and does not scale.   The role of the assumption is to avoid the suppression of all superpositions,  as it would follow with fixed '$w$' 
		in the mathematical limit of large $ N $ because of the factor $ N^2 $. This limit is a mathematical tool  and 
		our assumption is just a stratagem to prove theorems. We proved  that 
		\be\label{eq:criticalbeta}
		\text{	 there is a positive number $\beta_c=\frac{p^*(\epsilon)}{(\omega-1)\epsilon}$ such that, for $\beta > \beta_c$ and $ \omega>1 $, }\lim_{N \to \infty}\,[m^2]_{\beta} > \epsilon>0.
		\ee
		The factor $ p^*(\epsilon) $ is a large deviation rate functional related to the application of G\"artner-Ellis' theorem\cite{ellis}.
		The main observation used  to treat these ensembles  was to replace the random point on the sphere in \eqref{eq:Z} with 
		\be
		\phi_n \longrightarrow \frac{\phi_n}{\sqrt{\sum_n\,|\phi_n|^2}} ,
		\ee
		where $\{\,\phi_n: n=0,\dots, N\}$ are $N+1$ complex, or 2.$N$ real, numbers  
		distributed as i.i.d. standard (mean zero, norm one) Gaussians.   By this transformation	 the problem was converted into finding probabilities of rare events (called ``Large Deviations" theory \cite{ellis}) with the partition function becoming:
		\begin{equation}
			\no Z_N = c_N\,\int\,\prod_n\,d\chi_n\,\exp\left\{\,- ||\chi||^2/2 - E_N(\chi/||\chi||)\,
			\right\},\,\,\,\,
			||\chi||^2 = \underset{n}{\sum}|\chi_n|^2 .
		\end{equation}
		This observation allowed to exploit large deviation theory in proving magnetization by means of the following lemma: 
		
		\vspace{0.5cm}
		
		\ni Let $\Omega$ be a compact manifold without boundary, $ f $ a real-valued function on $\Omega$, $dx$ a finite
		measure on $\Omega$. Without loss of generality, we can take 
		
		\be
		\int_\Omega\,dx\, \= |\Omega| \= 1,
		\ee
		
		\ni and assume $f \geq 0$. 
		Let, for any bounded $g$ on $\Omega$:
		\bar
		\no [g] = \int_\Omega\,dx\,\exp\{ - f(x)\,\}\,g(x)/Z,\,\,\text{ where } Z = \int_\Omega\,dx\,\exp\{ - f(x)\,\}.
		\ear

		\begin{lemma}[ Concentration lemma] 
			
			Let there be two open subsets of $\Omega$, 
			called $U$ and $V$, and three positive numbers
			$\alpha$, $\eta$, and $\mu$ such that 
			\bar
			\no \hbox{(A)}&\ph\ph V \subset U\hspace{1.5cm} &
			\hbox{(B)}  \ph\ph f(x) \leq \eta, \ph\hbox{for}\ph x \in V; \\
			\no\hbox{(C)}&\ph\ph f(x) \geq  \alpha, \ph\hbox{for}\ph x \notin U, \hspace{1.5cm} & 
			\hbox{(D)}\ph\ph |V| \geq \,\mu.\\
			\ear
			Then: 	
			\be\label{eq:conclemma}
			[g] \= \frac{R+\xi}{1 + \zeta},
			\ee
			\ni where
			\bar
			\no R &= \int_U\,dx\,\exp\{ - f(x)\,\}\,g(x)/Z_U, &Z_U =\int_U\,dx\,\exp\{ - f(x)\,\};\\
			&\no|\xi| \leq e^{- \alpha}\,e^{\eta}\,\mu^{-1}\,||g||, &|\zeta| \leq e^{- \alpha}\,e^{\eta}\,\mu^{-1}.\\
			\ear
		\end{lemma}
		\ni Here $||g||$ denotes the supremum norm of $g$ on $\Omega$. 
		Note that $R \in \hbox{span}\{\,g(x):\,x \in U\,\}$.

		In the present  case $g(x)$ is  $m^2(\phi)$ and $U$  is the volume of positive magnetization. So the intuition behind this lemma is that  the measure concentrates in  a region of positive magnetization, but how it shrinks to zero in the thermodynamics limit has to be controlled. Hence the roles of the sets $U$ and $V$ and the 
		bounds on $f$. The volume $|V|$ should not be so small as to put a large factor in $\xi$ and $\zeta$; while $\alpha > \eta$. Then $\xi$ and
		$\zeta$ should tend to zero and the measure concentrates on the set $U$. The idea is that $V$ is a small neighborhood of the global
		minimum of $f$. The minimum may not occur at a single point, but on a subset.  
		
		We note that the ``balance of energy and entropy" game is contained in the difference	$\alpha - \eta$ and in $\mu$, measuring how $f$ increases 
		compared with the volume of the manifold that requires. 
		The proof proceeds as follow. 

		Assume  $\omega$ and $\epsilon$ positive numbers satisfy:
		
		\bar
		&& 1 < \omega < 4/3;\label{eq:rangeomega}\\
		&& \epsilon < \frac{1}{4}\,\left(\, 1 + \sqrt{1 - 4\,r}\,\right)^2;\label{epineq}\\
		&& r \= \frac{\omega - 1}{\omega}.
		\ear
		
		\ni (The bound $ 4/3 $ was  required for the application   of G\"artner-Ellis' theorem but  we do not expect that it has physical meaning.  With the specified range for $\omega$, (\ref{epineq}) always holds if $\epsilon < 1/4$)

		\def\ooN{{\frac{1}{N}}}
		\def\sdel{{\sqrt{\delta}}}
		\def\mhalf{{-\frac{1}{2}}}
		\def\Amax{{A_{\hbox{max}}}}
		\def\Amin{{A_{\hbox{min}}}}

		To apply  the Concentration Lemma. We defined  the quantities involved in the Lemma. For the sets $U$ and $V$ we took, given two positive numbers $\epsilon$ and $\eta$ ($\eta$ may depend on $N$),  
		\bar
		U \= \left\{\,\phi:\ph m^2(\phi)\,\geq\, \epsilon\,\right\};\,\,\,\, V \= \left\{\,f \leq \,\eta\,\right\},\,\, \alpha \= \beta\,N\,\left(\,1 - \epsilon\,\right)
		\ear
		The set $ V$ has the equivalent	description of set :
		\begin{equation}
			V = \left\{\,\phi:\ph\ph \omega\,m^2(\phi) \geq (\omega - 1)\,
			\sum\,|\phi_n|^2\,g_n^2\,
			+\,(1-\eta/\beta N)\,\right\} =\left\{\,\phi:\ph\ph m^2(\phi) \geq
			\sum\,|\phi_n|^2\,a_n
			\,\right\}
		\end{equation}
		where $ g_n = 1 - \frac{2\,n}{N} $ and 	$
		\no a_n \= \frac{\omega - 1}{\omega}\,g_n^2 \+ \frac{1-\eta'/\beta}{\omega}
		\no \= r\,g_n^2 \+ \delta.$
		Since $\omega > 1$, assuming $\eta = \eta'\,N$ and $1 - \eta'/\beta > \omega\,\epsilon$,
		the condition defining $V$ implies $V \subset U$.
		Since $\omega > 1$, it follows from $f \= N\,\beta\,\left\{\, 1 - m^2(\phi) + (\omega - 1)\,D(\phi)\,\right\}$ that $f \geq \alpha$ on $U^c$.
		For the lower bound on $|V|$, we translated to a model with i.i.d. Gaussians, call them
		$\chi_n$, replacing:
		
		\be
		\phi_n \longrightarrow\, \frac{\chi_n}{\sqrt{\sum |\chi_n|^2}},
		\ee
		
		\ni and the definition of $V$ becomes:
		
		\be
		V \= \left\{\,\left(\,\sumnN\,|\chi_n|^2\,g_n\,\right)^2 \geq \left[\,\sumnN\,a_n\,|\chi_n|^2\,
		\,\right]\,\sumnN\,|\chi_n|^2\,\right\}.
		\ee
		
		\ni 
		One has to find an exponential lower bound on $P[V]$. Since $a_n \geq \delta$,  a lower bound is 
		\bar
		\no && P\left[\,V\,\right] \geq  P\left[\,\left(\,\ooN\sumnN\,|\chi_n|^2\,g_n\,\right)^2 \geq 
		\left\{\,\ooN\sumnN\,a_n\,|\chi_n|^2\,
		\,\right\}\,\ooN\sumnN\,\frac{a_n}{\delta}\,|\chi_n|^2\,\right] \= \\
		\no && P\left[\,\sqrt{\delta}\,\ooN\sumnN\,|\chi_n|^2\,g_n\, \geq \ooN\sumnN\,a_n\,|\chi_n|^2\,
		\right] \+ P\left[\,\sqrt{\delta}\,\ooN\sumnN\,|\chi_n|^2\,g_n\, \leq - \ooN\sumnN\,a_n\,|\chi_n|^2\,
		\right].\\
		&&
		\ear 
		It is enough work on one the two probabilities, say  the first one.
		According to \Gartner-Ellis, we have to compute:
		\be
		p(\theta) \= \lim_{N\to\infty}\,\ooN\,\log\,{\cal{E}}\exp\{\,\theta\,\sum\,\chi_n^2\,b_n\,\},
		\ee
		\ni where $\cal{E}$ denotes expectation and $
		b_n \= \sdel\,g_n - a_n \= \sdel\,\left(\,1 - \frac{2n}{N}\,\right) - r
		\,\left(\,1 - \frac{2n}{N}\,\right)^2 - \delta.$
		
		Note that $b_n$ can take negative and positive values; hence, 
		$\theta$ must be restricted to ensure that the integral is finite. 
		From the standard Gaussian integral 	$
		p(\theta) \= \mhalf\, \lim_{N\to\infty}\,\sum\,\log\left( 1 - 2\,\theta\,b_n\,\right),
		$
		\ni which, provided the integrand is bounded, we recognize as the Reimann sum 
		converging to:
		$
		p(\theta) \= - \frac{1}{2}\,\int_0^1\,du\,\log\left( 1 - 2\,\theta\,A(1 - 2u)\,\right), \text{ where } A(x) \= \sdel\,x - 
		r\,x^2 \,\- \delta.
		$
		By a change of variable this equals:
		\be
		p(\theta) \= - \frac{1}{4}\,\int_{-1}^{1}\,dx\,\log\left( 1 - 2\,\theta\,A(x)\,\right).
		\ee
		The LD approach requires us to compute:
		\bar
		p^*(y) \= \sup_{\theta}\,\left\{\,\theta\,y - p(\theta)\,\right\}; \underline{p}^* \= \inf_{y > 0}\,p^*(y);
		\ear
		and then the asympotic lower bound is $\exp(-\underline{p}^*\,N\,)$, \cite{ellis}.  By the definition of the domain where $ p(\theta)<+\infty $ in G\"artner-Ellis' theorem, the supremum over $\theta$ in the definition of $p^*(y)$ can be limited to:
		\be
		\frac{1}{2\,\Amin} \leq \theta \leq \frac{1}{2\,\Amax},
		\ee
		\ni where $\Amin$ is negative but we don't need to know it, while by a simple computation:
		$
		\Amax \= \delta\,\left\{\,\frac{4 - 3\omega}{4(\omega -1)}\right\}.
		$
		\ni (This is the max on the whole line).
		We must have positive values of $A(x)$ somewhere in the interval [-1,1], for otherwise 
		$\lim\,p(\theta) = - \infty$ as $\theta \to \infty$, so $p^*(x) = +\infty$ for $x \geq 0$.
		This requires $\Amax >0$ and that the lower root of $A(x) = 0$ lies in the interval [0,1] (since
		$A(0) = - \delta < 0$ and $A'(0) > 0$).
		The root is easily computed to be:
		$
		x_{-} \= \frac{\sdel\,\left(\,1 - \sqrt{1 - 4\,r}\,\right)}{2\,r},
		$
		\ni Since $r < 1/4$,this root is real; 
		letting $u = \sqrt{1 - 4\,r}$ the condition $x_{-} < 1$ becomes 
		\be
		\delta < \frac{1}{4}\,\left(\, 1 + u\,\right)^2,
		\ee
		\ni since eventually  $\delta$ is identified with $\epsilon$,
		the above inequality is identical with (\ref{epineq}). Since $0\leq u \leq 1$, the infimum of the	right side is 1/4.
		
		The problem of computing $p^*(x)$ is then well defined, because logarithmic singularities are integrable. However, the derivative $p'(\theta)$ will go to infinity at the boundaries
		(Ellis calls such a function ``steep" and it is an assumption of his theorem).
		
		Assuming that $0 < \underline{p}^* < \infty$, it will suffice  to know that
		for some $\eta' > 0$:
		\bar
		\underline{p}^* + \eta' \leq \beta\,(\,1 - \epsilon\,);\,\,\,\, \eta' \leq \beta\,(\, 1 - \omega\,\epsilon\,); \label{ineq}
		\ear
		\ni One can define $\eta'$ to saturate the second inequality above 
		(note that $\omega\,\epsilon < 1$), which also yields
		$\delta = \epsilon$. 	The conclusion  \eqref{eq:criticalbeta} theorem follows. 

		\subsubsection*{Why indistinguishable spins and WFE were necessary }
		
		So the proof made the physical assumptions of having indistinguishable spins and the introduction of WFE.  We give  a brief mathematical explanation of why these were needed. 
		
		Without indistinguishability of the spins the phase space has dimension $2\times 2^N$ and when passing to  the Gaussian integral,
		\bar
		\no Z_N \= c_N\,\int\,\prod_S\,d\psi_S\,\exp\left\{\,- ||\psi||^2/2 - E(\psi/||\psi||)/2\,
		\right\};\,\,\,\
		||\psi||^2 \= \sumS.
		\ear
		because the components of $\psi$ are i.i.d. $N(0,1)$,   
		\be
		\sumS \approx a_N = 2\,2^N,
		\ee
		namely we have a ratio of  energies of order $N$ and $ ||\psi||^2 $ of order $2^N$, this kills regardless of the energy any interesting behavior and make the model a simple Gaussian integral.

		Without the WFE instead there was not symmetry breaking because the energy minima were given not only by ``all spins up'' and ``all spins down'' but also by their superposition  . This was shown by showing that, regardless of the temperature,
		\be
		\lim_{N\to \infty}\,\int_{m(\phi) \geq \epsilon}\,d\phi\,\exp\{\, -\,\beta\,E_{}\,\}/Z =  0, \,\,	Z \= \int\,d\phi\,\exp\{\, -\,\beta\,E_{}\,\}.
		\ee
		The integrals were rewritten as
		$
		\int_{m(\phi) \geq \epsilon}\,d\phi\,\exp\{\, -f\,\}/Z, \,\,\, Z = \int\,d\phi\,\exp\{\, -f \,\},
		$
		where
		$
		f = \beta\,N\,\left(\,1 - \sum\,|\phi_n|^2\,g_n^2\,\right);\,\,\,\, g_n = 1 - 2\,n/N.
		$
		
		\ni Where  I  added a term so that $0\leq \,f\,\leq \beta\,N$. So we had to evaluate the ratio of two probabilities: the probability of  the set	$B = \{m(\phi) \geq \epsilon\}$ in the numerator, and $B$ equal to the whole sphere in 	the denominator. Thus for the numerator was estimated:
		\be
		|\,\{ \beta\,N\,\left(\,1 - \sum\,|\phi_n|^2\,g_n^2\,\right) \leq \beta\,N\,x;\, 
		\sum |\phi_n|^2\,g_n \geq \,\epsilon\,\}\,|,
		\ee
		
		\ni which, replacing wavefunctions on the sphere by i.i.d. Gaussians, equivalently: 
		
		\be
		\Plb (1-x)\,\sum \chi_n^2 \leq \sum\,\chi_n^2\,g_n^2; \,\sum\,\chi_n^2\,g_n \geq 
		\epsilon\, \sum \chi_n^2 \Prb,
		\ee
		
		\ni which was the same as writing (introducing factors of $1/N$ ):

		\be
		\Plb 1/N\, \sum \chi_n^2 \,(1 - x - g_n^2) \leq 0; 1/N \,\sum\,\chi_n^2\,(\,g_n - \epsilon\,) 
		\geq 0 \Prb.\label{setB}
		\ee
		This last one has the interpretation of the conditional probability that the magnetization
		is greater than $\epsilon$, given a bound `$x$' on the energy. Again thanks to  the \Gartner-Ellis theorem I concluded that 
		\be
		\lim_{N \to \infty} p_N(x) = \frac{P_{2;N}(x)}{P_{1;N}(x)}=0.
		\ee
		where $P_{2;N}(x)$ stand for the probability of the set appearing above, 
		and $P_{1;N}(x)$ for the probability with the second restriction dropped.

		\section{Auxiliary proofs to DCI conditions}\label{app:chaos_cond}

		\begin{theorem}
			Let $M$ be an even-dimensional ($2\# \times 2\#$) matrix and
			\begin{equation}
				\det M < 0.
			\end{equation}
			Then $M$ has both positive and negative eigenvalues.
		\end{theorem}
		
		\begin{proof}
			
			The characteristic polynomial $p(\lambda)= \det(M -\lambda I)$ is of even order, with coefficient one, because $M$ has even dimension.  Since the leading term is  $\lambda^{2\#}$  it follows that  $\underset{\lambda \rightarrow \pm \infty}{\lim}p(\lambda)= +\infty$. The constant term is $\det M$; if the latter is negative, $p(\lambda)$ must cross zero at least once for a positive value and once for a negative one.
			
		\end{proof}
		Consider $P=0$; the dynamics on the tangent plane simplify as
		\begin{equation}\label{eq:Msimple}
			\frac{d}{dt} \begin{pmatrix} \xi \\ \eta \end{pmatrix} = M(t) \begin{pmatrix} \xi \\ \eta \end{pmatrix}
			\qquad \text{with} \quad
			M = \begin{pmatrix} A & B \\ C & D \end{pmatrix}
			\quad \text{where} \quad
			_{}	\begin{aligned}
				A &= 0 \\
				D &= 0 \\
				B &= \Lambda  \\
				C &= -\Lambda + v \otimes v
			\end{aligned}.
		\end{equation}

		\begin{lemma}
			For the system  \eqref{eq:Msimple},  the condition $\det M <0$  is equivalent to 
			\begin{equation}\label{eq:DCI2}
				v^t \Lambda^{-1} v > 1.
			\end{equation}
		\end{lemma}
		\begin{proof}
			
			By ``Schur complement'',  $\det M =(-1)^{\#} \det B \,\,\det (C - D B^{-1} A)= (-1)^{\#} \det B \det (C)= (-1)^{\#} \det \Lambda \,\,\det(-\Lambda + v\otimes v)$. Using Weinstein–Aronszajn identity,  $\det(-\Lambda + v\otimes v)= \det(-\Lambda)[1-v^t \Lambda^{-1} \, v]$ , from which follows that 
			\begin{equation}
				\det M= (\det \Lambda)^2 \left[1-v^t \Lambda^{-1} \, v\right],
			\end{equation}
			and so $\det M <0 \iff v^t \Lambda^{-1} \, v > 1 $.
			
		\end{proof}

		\begin{acknowledgments}
			\ni I  thanks William David Wick for the time spent during  the discussions and electronic communications.  I thank CINECA where the first preprint was written in the spare time. I thank Michele Campisi for indicating the references \cite{Anza,Brody}.

			Leonardo De Carlo is a member of the 'Meccanica dei sistemi discreti' section of the Gruppo Nazionale per la Fisica Matematica-Instituto Nazionale di Alta Matematica  \href{https://www.altamatematica.it/gnfm/en/aderenti/aderenti-2024/}{(GNFM-INdAM)}.
			
			This work was partly supported by the Portuguese Science and Technology Foundation FCT, via the research centre GFM, references UID/00208/2025 (\url{https://doi.org/10.54499/UID/00208/2025}) and UID/PRR/00208/2025 (\url{https://doi.org/10.54499/UID/PRR/00208/2025}).
		\end{acknowledgments}
		
		\section*{COI}
		
		\ni The author have no conflict of interest.
		
		\section*{Data Availability Statement}
		
		\ni The data that support the findings of this study are openly available in ArXiv at \cite{MP3,MP4,ContinuumChaos,DotonScreen}.

		\nocite{*}
		\bibliography{aipsamp}

\providecommand{\noopsort}[1]{}\providecommand{\singleletter}[1]{#1}%
\begin{thebibliography}{40}%
\makeatletter
\providecommand \@ifxundefined [1]{%
 \@ifx{#1\undefined}
}%
\providecommand \@ifnum [1]{%
 \ifnum #1\expandafter \@firstoftwo
 \else \expandafter \@secondoftwo
 \fi
}%
\providecommand \@ifx [1]{%
 \ifx #1\expandafter \@firstoftwo
 \else \expandafter \@secondoftwo
 \fi
}%
\providecommand \natexlab [1]{#1}%
\providecommand \enquote  [1]{``#1''}%
\providecommand \bibnamefont  [1]{#1}%
\providecommand \bibfnamefont [1]{#1}%
\providecommand \citenamefont [1]{#1}%
\providecommand \href@noop [0]{\@secondoftwo}%
\providecommand \href [0]{\begingroup \@sanitize@url \@href}%
\providecommand \@href[1]{\@@startlink{#1}\@@href}%
\providecommand \@@href[1]{\endgroup#1\@@endlink}%
\providecommand \@sanitize@url [0]{\catcode `\\12\catcode `\$12\catcode
  `\&12\catcode `\#12\catcode `\^12\catcode `\_12\catcode `\%12\relax}%
\providecommand \@@startlink[1]{}%
\providecommand \@@endlink[0]{}%
\providecommand \url  [0]{\begingroup\@sanitize@url \@url }%
\providecommand \@url [1]{\endgroup\@href {#1}{\urlprefix }}%
\providecommand \urlprefix  [0]{URL }%
\providecommand \Eprint [0]{\href }%
\providecommand \doibase [0]{http://dx.doi.org/}%
\providecommand \selectlanguage [0]{\@gobble}%
\providecommand \bibinfo  [0]{\@secondoftwo}%
\providecommand \bibfield  [0]{\@secondoftwo}%
\providecommand \translation [1]{[#1]}%
\providecommand \BibitemOpen [0]{}%
\providecommand \bibitemStop [0]{}%
\providecommand \bibitemNoStop [0]{.\EOS\space}%
\providecommand \EOS [0]{\spacefactor3000\relax}%
\providecommand \BibitemShut  [1]{\csname bibitem#1\endcsname}%
\let\auto@bib@innerbib\@empty
\bibitem [{\citenamefont {Arndt}\ and\ \citenamefont
  {Hornberger}(2014)}]{ArnHorn}%
  \BibitemOpen
  \bibfield  {author} {\bibinfo {author} {\bibfnamefont {M.}~\bibnamefont
  {Arndt}}\ and\ \bibinfo {author} {\bibfnamefont {K.}~\bibnamefont
  {Hornberger}},\ }\bibfield  {title} {\enquote {\bibinfo {title} {Testing the
  limits of quantum mechanical superpositions},}\ }\href@noop {} {\bibfield
  {journal} {\bibinfo  {journal} {Nat. Phys.}\ }\textbf {\bibinfo {volume}
  {10}},\ \bibinfo {pages} {271--277} (\bibinfo {year} {2014})}\BibitemShut
  {NoStop}%
\bibitem [{\citenamefont {{ERC Advanced Grant}}(2025)}]{Q-Tube}%
  \BibitemOpen
  \bibfield  {author} {\bibinfo {author} {\bibnamefont {{ERC Advanced
  Grant}}},\ }\href@noop {} {\enquote {\bibinfo {title} {Q-tube},}\ }\bibinfo
  {howpublished}
  {\url{https://bist.eu/icfo-researcher-awarded-erc-advanced-grant/}} (\bibinfo
  {year} {2025})\BibitemShut {NoStop}%
\bibitem [{\citenamefont {{ERC Synergy Grant}}(2023)}]{Q-Xtreme}%
  \BibitemOpen
  \bibfield  {author} {\bibinfo {author} {\bibnamefont {{ERC Synergy Grant}}},\
  }\href@noop {} {\enquote {\bibinfo {title} {Q-xtreme},}\ }\bibinfo
  {howpublished} {\url{https://cordis.europa.eu/project/id/951234}} (\bibinfo
  {year} {2023}),\ \bibinfo {note} {grant agreement ID: 951234}\BibitemShut
  {NoStop}%
\bibitem [{\citenamefont {Bassi}, \citenamefont {Dorato},\ and\ \citenamefont
  {Ulbricht}(2023)}]{Collapse}%
  \BibitemOpen
  \bibfield  {author} {\bibinfo {author} {\bibfnamefont {A.}~\bibnamefont
  {Bassi}}, \bibinfo {author} {\bibfnamefont {M.}~\bibnamefont {Dorato}}, \
  and\ \bibinfo {author} {\bibfnamefont {H.}~\bibnamefont {Ulbricht}},\
  }\bibfield  {title} {\enquote {\bibinfo {title} {Collapse models: A
  theoretical, experimental and philosophical review},}\ }\href@noop {}
  {\bibfield  {journal} {\bibinfo  {journal} {Entropy}\ }\textbf {\bibinfo
  {volume} {25}},\ \bibinfo {pages} {645} (\bibinfo {year} {2023})}\BibitemShut
  {NoStop}%
\bibitem [{\citenamefont {Hance}\ and\ \citenamefont
  {Hossenfelder}(2022)}]{HanceHossenfelder2022}%
  \BibitemOpen
  \bibfield  {author} {\bibinfo {author} {\bibfnamefont {J.~R.}\ \bibnamefont
  {Hance}}\ and\ \bibinfo {author} {\bibfnamefont {S.}~\bibnamefont
  {Hossenfelder}},\ }\bibfield  {title} {\enquote {\bibinfo {title} {What does
  it take to solve the measurement problem?}}\ }\href@noop {} {\bibfield
  {journal} {\bibinfo  {journal} {J. Phys. Commun.}\ }\textbf {\bibinfo
  {volume} {6}},\ \bibinfo {pages} {102001} (\bibinfo {year}
  {2022})}\BibitemShut {NoStop}%
\bibitem [{\citenamefont {Wick}(2017)}]{MP1}%
  \BibitemOpen
  \bibfield  {author} {\bibinfo {author} {\bibfnamefont {W.~D.}\ \bibnamefont
  {Wick}},\ }\href@noop {} {\enquote {\bibinfo {title} {On the non-linear
  quantum mechanics and the measurement problem i. blocking cats},}\ }\bibinfo
  {howpublished} {arXiv:1710.03278} (\bibinfo {year} {2017})\BibitemShut
  {NoStop}%
\bibitem [{\citenamefont {De~Carlo}\ and\ \citenamefont
  {Wick}(2023)}]{DeCarlo-Wick}%
  \BibitemOpen
  \bibfield  {author} {\bibinfo {author} {\bibfnamefont {L.}~\bibnamefont
  {De~Carlo}}\ and\ \bibinfo {author} {\bibfnamefont {W.~D.}\ \bibnamefont
  {Wick}},\ }\bibfield  {title} {\enquote {\bibinfo {title} {On magnetic models
  in wavefunction ensembles},}\ }\href@noop {} {\bibfield  {journal} {\bibinfo
  {journal} {Entropy}\ }\textbf {\bibinfo {volume} {25}},\ \bibinfo {pages}
  {564} (\bibinfo {year} {2023})}\BibitemShut {NoStop}%
\bibitem [{\citenamefont {Leggett}(1980)}]{Leggett}%
  \BibitemOpen
  \bibfield  {author} {\bibinfo {author} {\bibfnamefont {A.~J.}\ \bibnamefont
  {Leggett}},\ }\bibfield  {title} {\enquote {\bibinfo {title} {Macroscopic
  quantum systems and the quantum theory of measurement},}\ }\href@noop {}
  {\bibfield  {journal} {\bibinfo  {journal} {Supplement to the Progress of
  Theoretical Physics}\ }\textbf {\bibinfo {volume} {69}},\ \bibinfo {pages}
  {90} (\bibinfo {year} {1980})}\BibitemShut {NoStop}%
\bibitem [{\citenamefont {Schrödinger}(1952)}]{ES1}%
  \BibitemOpen
  \bibfield  {author} {\bibinfo {author} {\bibfnamefont {E.}~\bibnamefont
  {Schrödinger}},\ }\href@noop {} {\emph {\bibinfo {title} {Statistical
  Thermodynamics}}}\ (\bibinfo  {publisher} {Dover Publications, Inc.},\
  \bibinfo {address} {New York},\ \bibinfo {year} {1952})\BibitemShut {NoStop}%
\bibitem [{\citenamefont {Schrödinger}(1928)}]{ES2}%
  \BibitemOpen
  \bibfield  {author} {\bibinfo {author} {\bibfnamefont {E.}~\bibnamefont
  {Schrödinger}},\ }\enquote {\bibinfo {title} {The exchange of energy
  according to wave-mechanics},}\ in\ \href@noop {} {\emph {\bibinfo
  {booktitle} {Collected Papers on Wave Mechanics}}}\ (\bibinfo  {publisher}
  {Blackie \& Son Limited},\ \bibinfo {address} {London and Glasgow},\ \bibinfo
  {year} {1928})\ pp.\ \bibinfo {pages} {137--146}\BibitemShut {NoStop}%
\bibitem [{\citenamefont {Bloch}(2000)}]{Bloch}%
  \BibitemOpen
  \bibfield  {author} {\bibinfo {author} {\bibfnamefont {F.}~\bibnamefont
  {Bloch}},\ }\href@noop {} {\emph {\bibinfo {title} {Fundamentals of
  Statistical Mechanics, prepared by J.D. Walecka}}}\ (\bibinfo  {publisher}
  {Imperial College Press, World Scientific},\ \bibinfo {year}
  {2000})\BibitemShut {NoStop}%
\bibitem [{\citenamefont {Campisi}(2013)}]{Campisi}%
  \BibitemOpen
  \bibfield  {author} {\bibinfo {author} {\bibfnamefont {M.}~\bibnamefont
  {Campisi}},\ }\bibfield  {title} {\enquote {\bibinfo {title} {Quantum
  fluctuation relations for ensembles of wave functions},}\ }\href@noop {}
  {\bibfield  {journal} {\bibinfo  {journal} {New Journal of Physics}\ }\textbf
  {\bibinfo {volume} {15}},\ \bibinfo {pages} {115008} (\bibinfo {year}
  {2013})}\BibitemShut {NoStop}%
\bibitem [{\citenamefont {Jona-Lasinio}\ and\ \citenamefont
  {Presilla}(2006)}]{Jona-Presilla}%
  \BibitemOpen
  \bibfield  {author} {\bibinfo {author} {\bibfnamefont {G.}~\bibnamefont
  {Jona-Lasinio}}\ and\ \bibinfo {author} {\bibfnamefont {C.}~\bibnamefont
  {Presilla}},\ }\bibfield  {title} {\enquote {\bibinfo {title} {On the
  statistics of quantum expectations for systems in thermal equilibrium},}\
  }in\ \href@noop {} {\emph {\bibinfo {booktitle} {Quantum Mechanics: Are there
  Quantum Jumps? and On the Present Status of Quantum Mechanics}}},\ Vol.\
  \bibinfo {volume} {844},\ \bibinfo {editor} {edited by\ \bibinfo {editor}
  {\bibfnamefont {A.}~\bibnamefont {Bassi}}, \bibinfo {editor} {\bibfnamefont
  {D.}~\bibnamefont {Dürr}}, \bibinfo {editor} {\bibfnamefont
  {T.}~\bibnamefont {Weber}}, \ and\ \bibinfo {editor} {\bibfnamefont
  {N.}~\bibnamefont {Zanghì}}}\ (\bibinfo  {publisher} {American Institute of
  Physics},\ \bibinfo {address} {Melville, NY},\ \bibinfo {year} {2006})\ pp.\
  \bibinfo {pages} {200--205}\BibitemShut {NoStop}%
\bibitem [{\citenamefont {Lebowitz}(2021)}]{lebowitz}%
  \BibitemOpen
  \bibfield  {author} {\bibinfo {author} {\bibfnamefont {J.}~\bibnamefont
  {Lebowitz}},\ }\href@noop {} {\enquote {\bibinfo {title} {Microscopic origin
  of macroscopic behavior},}\ }\bibinfo {howpublished} {arXiv:2105.03470}
  (\bibinfo {year} {2021})\BibitemShut {NoStop}%
\bibitem [{\citenamefont {Anza}\ and\ \citenamefont
  {Crutchfield}(2022)}]{Anza}%
  \BibitemOpen
  \bibfield  {author} {\bibinfo {author} {\bibfnamefont {F.}~\bibnamefont
  {Anza}}\ and\ \bibinfo {author} {\bibfnamefont {J.~P.}\ \bibnamefont
  {Crutchfield}},\ }\bibfield  {title} {\enquote {\bibinfo {title} {Geometric
  quantum thermodynamics},}\ }\href@noop {} {\bibfield  {journal} {\bibinfo
  {journal} {Phys. Rev. E}\ }\textbf {\bibinfo {volume} {106}},\ \bibinfo
  {pages} {054102} (\bibinfo {year} {2022})}\BibitemShut {NoStop}%
\bibitem [{\citenamefont {Brody}\ and\ \citenamefont {Hughston}(1998)}]{Brody}%
  \BibitemOpen
  \bibfield  {author} {\bibinfo {author} {\bibfnamefont {D.~C.}\ \bibnamefont
  {Brody}}\ and\ \bibinfo {author} {\bibfnamefont {L.~P.}\ \bibnamefont
  {Hughston}},\ }\bibfield  {title} {\enquote {\bibinfo {title} {The quantum
  canonical ensemble},}\ }\href@noop {} {\bibfield  {journal} {\bibinfo
  {journal} {J. Math. Phys.}\ }\textbf {\bibinfo {volume} {39}},\ \bibinfo
  {pages} {6502--6508} (\bibinfo {year} {1998})}\BibitemShut {NoStop}%
\bibitem [{\citenamefont {Wick}(2018)}]{MP3}%
  \BibitemOpen
  \bibfield  {author} {\bibinfo {author} {\bibfnamefont {W.~D.}\ \bibnamefont
  {Wick}},\ }\href@noop {} {\enquote {\bibinfo {title} {On the non-linear
  quantum mechanics and the measurement problem iii. poincaré},}\ }\bibinfo
  {howpublished} {arXiv:1803.11236} (\bibinfo {year} {2018})\BibitemShut
  {NoStop}%
\bibitem [{\citenamefont {Wick}(2025{\natexlab{a}})}]{ContinuumChaos}%
  \BibitemOpen
  \bibfield  {author} {\bibinfo {author} {\bibfnamefont {W.~D.}\ \bibnamefont
  {Wick}},\ }\href@noop {} {\enquote {\bibinfo {title} {Chaos in a nonlinear
  wavefunction model: An alternative to born's probability hypothesis},}\
  }\bibinfo {howpublished} {arXiv:2502.02698} (\bibinfo {year}
  {2025}{\natexlab{a}})\BibitemShut {NoStop}%
\bibitem [{\citenamefont {Abdi}\ \emph {et~al.}(2016)\citenamefont {Abdi},
  \citenamefont {Degenfeld-Schonburg}, \citenamefont {Sameti}, \citenamefont
  {Navarrete-Benlloch},\ and\ \citenamefont {Hartmann}}]{Abdietl}%
  \BibitemOpen
  \bibfield  {author} {\bibinfo {author} {\bibfnamefont {M.}~\bibnamefont
  {Abdi}}, \bibinfo {author} {\bibfnamefont {P.}~\bibnamefont
  {Degenfeld-Schonburg}}, \bibinfo {author} {\bibfnamefont {M.}~\bibnamefont
  {Sameti}}, \bibinfo {author} {\bibfnamefont {C.}~\bibnamefont
  {Navarrete-Benlloch}}, \ and\ \bibinfo {author} {\bibfnamefont {M.~J.}\
  \bibnamefont {Hartmann}},\ }\bibfield  {title} {\enquote {\bibinfo {title}
  {Dissipative optomechanical preparation of macroscopic quantum superposition
  states},}\ }\href@noop {} {\bibfield  {journal} {\bibinfo  {journal} {Phys.
  Rev. Lett.}\ }\textbf {\bibinfo {volume} {116}},\ \bibinfo {pages} {233604}
  (\bibinfo {year} {2016})}\BibitemShut {NoStop}%
\bibitem [{\citenamefont {Wick}(2025{\natexlab{b}})}]{DotonScreen}%
  \BibitemOpen
  \bibfield  {author} {\bibinfo {author} {\bibfnamefont {W.~D.}\ \bibnamefont
  {Wick}},\ }\href@noop {} {\enquote {\bibinfo {title} {That dot on the screen:
  also, what about born? and other objections to wavefunction physics},}\
  }\bibinfo {howpublished} {arXiv:2504.17808} (\bibinfo {year}
  {2025}{\natexlab{b}})\BibitemShut {NoStop}%
\bibitem [{\citenamefont {Press}\ \emph {et~al.}(1988)\citenamefont {Press},
  \citenamefont {Teukolsky}, \citenamefont {Vetterling},\ and\ \citenamefont
  {Flannery}}]{NumRec}%
  \BibitemOpen
  \bibfield  {author} {\bibinfo {author} {\bibfnamefont {W.~H.}\ \bibnamefont
  {Press}}, \bibinfo {author} {\bibfnamefont {S.~A.}\ \bibnamefont
  {Teukolsky}}, \bibinfo {author} {\bibfnamefont {W.~T.}\ \bibnamefont
  {Vetterling}}, \ and\ \bibinfo {author} {\bibfnamefont {B.~P.}\ \bibnamefont
  {Flannery}},\ }\href@noop {} {\emph {\bibinfo {title} {Numerical Recipes in
  C}}}\ (\bibinfo  {publisher} {Cambridge University Press},\ \bibinfo
  {address} {Cambridge, UK},\ \bibinfo {year} {1988})\BibitemShut {NoStop}%
\bibitem [{\citenamefont {Tao}(2016)}]{Tao}%
  \BibitemOpen
  \bibfield  {author} {\bibinfo {author} {\bibfnamefont {M.}~\bibnamefont
  {Tao}},\ }\bibfield  {title} {\enquote {\bibinfo {title} {Explicit symplectic
  approximation of nonseparable hamiltonians: algorithm and long-time
  performance},}\ }\href@noop {} {\bibfield  {journal} {\bibinfo  {journal}
  {Phys. Rev. E}\ }\textbf {\bibinfo {volume} {94}},\ \bibinfo {pages} {043303}
  (\bibinfo {year} {2016})}\BibitemShut {NoStop}%
\bibitem [{\citenamefont {Wick}(2019)}]{MP4}%
  \BibitemOpen
  \bibfield  {author} {\bibinfo {author} {\bibfnamefont {W.~D.}\ \bibnamefont
  {Wick}},\ }\href@noop {} {\enquote {\bibinfo {title} {On non-linear quantum
  mechanics and the measurement problem iv. experimental tests},}\ }\bibinfo
  {howpublished} {arXiv:1908.02352} (\bibinfo {year} {2019})\BibitemShut
  {NoStop}%
\bibitem [{\citenamefont {Ashkin}\ \emph {et~al.}(1986)\citenamefont {Ashkin},
  \citenamefont {Dziedzic}, \citenamefont {Bjorkholm},\ and\ \citenamefont
  {Chu}}]{Ashkin1986}%
  \BibitemOpen
  \bibfield  {author} {\bibinfo {author} {\bibfnamefont {A.}~\bibnamefont
  {Ashkin}}, \bibinfo {author} {\bibfnamefont {J.~M.}\ \bibnamefont
  {Dziedzic}}, \bibinfo {author} {\bibfnamefont {J.~E.}\ \bibnamefont
  {Bjorkholm}}, \ and\ \bibinfo {author} {\bibfnamefont {S.}~\bibnamefont
  {Chu}},\ }\bibfield  {title} {\enquote {\bibinfo {title} {Observation of a
  single-beam gradient force optical trap for dielectric particles},}\
  }\href@noop {} {\bibfield  {journal} {\bibinfo  {journal} {Opt. Lett.}\
  }\textbf {\bibinfo {volume} {11}},\ \bibinfo {pages} {288--290} (\bibinfo
  {year} {1986})}\BibitemShut {NoStop}%
\bibitem [{\citenamefont {Higbie}\ and\ \citenamefont
  {Stamper-Kurn}(2004)}]{HigbieStamper}%
  \BibitemOpen
  \bibfield  {author} {\bibinfo {author} {\bibfnamefont {J.}~\bibnamefont
  {Higbie}}\ and\ \bibinfo {author} {\bibfnamefont {D.~M.}\ \bibnamefont
  {Stamper-Kurn}},\ }\bibfield  {title} {\enquote {\bibinfo {title} {Generating
  macroscopic-quantum-superposition states in momentum and internal-state space
  from bose-einstein condensates with repulsive interactions},}\ }\href@noop {}
  {\bibfield  {journal} {\bibinfo  {journal} {Phys. Rev. A}\ }\textbf {\bibinfo
  {volume} {69}},\ \bibinfo {pages} {053605} (\bibinfo {year}
  {2004})}\BibitemShut {NoStop}%
\bibitem [{\citenamefont {Sadler}\ \emph {et~al.}()\citenamefont {Sadler},
  \citenamefont {Higbie}, \citenamefont {Hetherington}, \citenamefont {Schmid},
  \citenamefont {Pasienski},\ and\ \citenamefont {Stamper-Kurn}}]{Sadler2004}%
  \BibitemOpen
  \bibfield  {author} {\bibinfo {author} {\bibfnamefont {L.}~\bibnamefont
  {Sadler}}, \bibinfo {author} {\bibfnamefont {J.}~\bibnamefont {Higbie}},
  \bibinfo {author} {\bibfnamefont {C.}~\bibnamefont {Hetherington}}, \bibinfo
  {author} {\bibfnamefont {S.}~\bibnamefont {Schmid}}, \bibinfo {author}
  {\bibfnamefont {M.}~\bibnamefont {Pasienski}}, \ and\ \bibinfo {author}
  {\bibfnamefont {D.~M.}\ \bibnamefont {Stamper-Kurn}},\ }\enquote {\bibinfo
  {title} {Periodically-dressed {Bose-Einstein} condensate in $^{87}${Rb}},}\
  \BibitemShut {NoStop}%
\bibitem [{\citenamefont {Ghirardi}, \citenamefont {Rimini},\ and\
  \citenamefont {Weber}(1986)}]{GRW86}%
  \BibitemOpen
\bibfield  {title} {  }\bibfield  {author} {\bibinfo {author} {\bibfnamefont
  {G.~C.}\ \bibnamefont {Ghirardi}}, \bibinfo {author} {\bibfnamefont
  {A.}~\bibnamefont {Rimini}}, \ and\ \bibinfo {author} {\bibfnamefont
  {T.}~\bibnamefont {Weber}},\ }\bibfield  {title} {\enquote {\bibinfo {title}
  {Unified dynamics for microscopic and macroscopic systems},}\ }\href@noop {}
  {\bibfield  {journal} {\bibinfo  {journal} {Phys. Rev. D}\ }\textbf {\bibinfo
  {volume} {34}},\ \bibinfo {pages} {470} (\bibinfo {year} {1986})}\BibitemShut
  {NoStop}%
\bibitem [{\citenamefont {Ghirardi}, \citenamefont {Pearle},\ and\
  \citenamefont {Rimini}(1990)}]{GPR90}%
  \BibitemOpen
  \bibfield  {author} {\bibinfo {author} {\bibfnamefont {G.~C.}\ \bibnamefont
  {Ghirardi}}, \bibinfo {author} {\bibfnamefont {P.}~\bibnamefont {Pearle}}, \
  and\ \bibinfo {author} {\bibfnamefont {A.}~\bibnamefont {Rimini}},\
  }\bibfield  {title} {\enquote {\bibinfo {title} {Markov processes in hilbert
  space and continuous spontaneous localization of systems of identical
  particles},}\ }\href@noop {} {\bibfield  {journal} {\bibinfo  {journal}
  {Phys. Rev. A}\ }\textbf {\bibinfo {volume} {42}},\ \bibinfo {pages} {78}
  (\bibinfo {year} {1990})}\BibitemShut {NoStop}%
\bibitem [{\citenamefont {Fu}(1997)}]{Fu97}%
  \BibitemOpen
  \bibfield  {author} {\bibinfo {author} {\bibfnamefont {Q.}~\bibnamefont
  {Fu}},\ }\bibfield  {title} {\enquote {\bibinfo {title} {Spontaneous
  radiation of free electrons in a nonrelativistic collapse model},}\
  }\href@noop {} {\bibfield  {journal} {\bibinfo  {journal} {Phys. Rev. A}\
  }\textbf {\bibinfo {volume} {56}},\ \bibinfo {pages} {1806} (\bibinfo {year}
  {1997})}\BibitemShut {NoStop}%
\bibitem [{\citenamefont {Carlesso}\ \emph {et~al.}(2022)\citenamefont
  {Carlesso}, \citenamefont {Donadi}, \citenamefont {Ferialdi}, \citenamefont
  {Paternostro}, \citenamefont {Ulbricht},\ and\ \citenamefont
  {Bassi}}]{Carlesso2022}%
  \BibitemOpen
  \bibfield  {author} {\bibinfo {author} {\bibfnamefont {M.}~\bibnamefont
  {Carlesso}}, \bibinfo {author} {\bibfnamefont {S.}~\bibnamefont {Donadi}},
  \bibinfo {author} {\bibfnamefont {L.}~\bibnamefont {Ferialdi}}, \bibinfo
  {author} {\bibfnamefont {M.}~\bibnamefont {Paternostro}}, \bibinfo {author}
  {\bibfnamefont {H.}~\bibnamefont {Ulbricht}}, \ and\ \bibinfo {author}
  {\bibfnamefont {A.}~\bibnamefont {Bassi}},\ }\bibfield  {title} {\enquote
  {\bibinfo {title} {Present status and future challenges of
  non-interferometric tests of collapse models},}\ }\href@noop {} {\bibfield
  {journal} {\bibinfo  {journal} {Nat. Phys.}\ }\textbf {\bibinfo {volume}
  {18}},\ \bibinfo {pages} {243--250} (\bibinfo {year} {2022})}\BibitemShut
  {NoStop}%
\bibitem [{\citenamefont {Bilardello}\ \emph {et~al.}(2016)\citenamefont
  {Bilardello}, \citenamefont {Donadi}, \citenamefont {Vinante},\ and\
  \citenamefont {Bassi}}]{Bilardello2016}%
  \BibitemOpen
  \bibfield  {author} {\bibinfo {author} {\bibfnamefont {M.}~\bibnamefont
  {Bilardello}}, \bibinfo {author} {\bibfnamefont {S.}~\bibnamefont {Donadi}},
  \bibinfo {author} {\bibfnamefont {A.}~\bibnamefont {Vinante}}, \ and\
  \bibinfo {author} {\bibfnamefont {A.}~\bibnamefont {Bassi}},\ }\bibfield
  {title} {\enquote {\bibinfo {title} {Bounds on collapse models from cold-atom
  experiments},}\ }\href@noop {} {\bibfield  {journal} {\bibinfo  {journal}
  {Physica A}\ }\textbf {\bibinfo {volume} {462}},\ \bibinfo {pages} {764--782}
  (\bibinfo {year} {2016})}\BibitemShut {NoStop}%
\bibitem [{\citenamefont {Kovachy}\ \emph {et~al.}(2015)\citenamefont
  {Kovachy}, \citenamefont {Hogan}, \citenamefont {Sugarbaker}, \citenamefont
  {Dickerson}, \citenamefont {Donnelly}, \citenamefont {Overstreet},\ and\
  \citenamefont {Kasevich}}]{Kovachy2015}%
  \BibitemOpen
  \bibfield  {author} {\bibinfo {author} {\bibfnamefont {T.}~\bibnamefont
  {Kovachy}}, \bibinfo {author} {\bibfnamefont {J.~M.}\ \bibnamefont {Hogan}},
  \bibinfo {author} {\bibfnamefont {A.}~\bibnamefont {Sugarbaker}}, \bibinfo
  {author} {\bibfnamefont {S.~M.}\ \bibnamefont {Dickerson}}, \bibinfo {author}
  {\bibfnamefont {C.~A.}\ \bibnamefont {Donnelly}}, \bibinfo {author}
  {\bibfnamefont {C.}~\bibnamefont {Overstreet}}, \ and\ \bibinfo {author}
  {\bibfnamefont {M.~A.}\ \bibnamefont {Kasevich}},\ }\bibfield  {title}
  {\enquote {\bibinfo {title} {Matter wave lensing to picokelvin
  temperatures},}\ }\href@noop {} {\bibfield  {journal} {\bibinfo  {journal}
  {Phys. Rev. Lett.}\ }\textbf {\bibinfo {volume} {114}},\ \bibinfo {pages}
  {143004} (\bibinfo {year} {2015})}\BibitemShut {NoStop}%
\bibitem [{\citenamefont {Wick}\ and\ \citenamefont
  {De~Carlo}(2025)}]{ContinuumChaosIslands}%
  \BibitemOpen
  \bibfield  {author} {\bibinfo {author} {\bibfnamefont {W.~D.}\ \bibnamefont
  {Wick}}\ and\ \bibinfo {author} {\bibfnamefont {L.}~\bibnamefont
  {De~Carlo}},\ }\href@noop {} {\enquote {\bibinfo {title} {In preparation
  update of: Islands of instability in nonlinear wavefunction models in the
  continuum: A different route to "chaos"},}\ }\bibinfo {howpublished}
  {arXiv:2512.09109} (\bibinfo {year} {2025})\BibitemShut {NoStop}%
\bibitem [{\citenamefont {Dirac}, \citenamefont {Fock},\ and\ \citenamefont
  {Podolsky}(1932)}]{Dirac1932}%
  \BibitemOpen
  \bibfield  {author} {\bibinfo {author} {\bibfnamefont {P.~A.~M.}\
  \bibnamefont {Dirac}}, \bibinfo {author} {\bibfnamefont {V.~A.}\ \bibnamefont
  {Fock}}, \ and\ \bibinfo {author} {\bibfnamefont {B.}~\bibnamefont
  {Podolsky}},\ }\bibfield  {title} {\enquote {\bibinfo {title} {On quantum
  electrodynamics},}\ }\href@noop {} {\bibfield  {journal} {\bibinfo  {journal}
  {Phys. Z. Sowjetunion}\ }\textbf {\bibinfo {volume} {2}},\ \bibinfo {pages}
  {468--479} (\bibinfo {year} {1932})},\ \bibinfo {note} {reprinted in:
  \textit{Selected Papers on Quantum Electrodynamics}, J. Schwinger, ed.,
  Dover, New York (1958)}\BibitemShut {NoStop}%
\bibitem [{\citenamefont {Lienert}, \citenamefont {Petrat},\ and\ \citenamefont
  {Tumulka}(2017)}]{Lienert2017}%
  \BibitemOpen
  \bibfield  {author} {\bibinfo {author} {\bibfnamefont {M.}~\bibnamefont
  {Lienert}}, \bibinfo {author} {\bibfnamefont {S.}~\bibnamefont {Petrat}}, \
  and\ \bibinfo {author} {\bibfnamefont {R.}~\bibnamefont {Tumulka}},\
  }\bibfield  {title} {\enquote {\bibinfo {title} {Multi-time wave
  functions},}\ }\href@noop {} {\bibfield  {journal} {\bibinfo  {journal} {J.
  Phys.: Conf. Ser.}\ }\textbf {\bibinfo {volume} {880}},\ \bibinfo {pages}
  {012006} (\bibinfo {year} {2017})}\BibitemShut {NoStop}%
\bibitem [{\citenamefont {Wick}(2020)}]{compatibility}%
  \BibitemOpen
  \bibfield  {author} {\bibinfo {author} {\bibfnamefont {W.~D.}\ \bibnamefont
  {Wick}},\ }\href@noop {} {\enquote {\bibinfo {title} {On non-linear quantum
  mechanics, space-time wavefunctions, and compatibility with general
  relativity},}\ }\bibinfo {howpublished} {arXiv:2008.08663} (\bibinfo {year}
  {2020})\BibitemShut {NoStop}%
\bibitem [{\citenamefont {Bassi}\ and\ \citenamefont
  {Ghirardi}(2003)}]{Bassi2003}%
  \BibitemOpen
  \bibfield  {author} {\bibinfo {author} {\bibfnamefont {A.}~\bibnamefont
  {Bassi}}\ and\ \bibinfo {author} {\bibfnamefont {G.}~\bibnamefont
  {Ghirardi}},\ }\bibfield  {title} {\enquote {\bibinfo {title} {Dynamical
  reduction models},}\ }\href@noop {} {\bibfield  {journal} {\bibinfo
  {journal} {Phys. Rep.}\ }\textbf {\bibinfo {volume} {379}},\ \bibinfo {pages}
  {257--426} (\bibinfo {year} {2003})}\BibitemShut {NoStop}%
\bibitem [{\citenamefont {Barut}(1991)}]{Barut1991}%
  \BibitemOpen
  \bibfield  {author} {\bibinfo {author} {\bibfnamefont {A.~O.}\ \bibnamefont
  {Barut}},\ }\bibfield  {title} {\enquote {\bibinfo {title} {Brief history and
  recent developments in electron theory and quantum electrodynamics},}\ }in\
  \href@noop {} {\emph {\bibinfo {booktitle} {The Electron}}},\ \bibinfo
  {editor} {edited by\ \bibinfo {editor} {\bibfnamefont {D.}~\bibnamefont
  {Hestenes}}\ and\ \bibinfo {editor} {\bibfnamefont {A.}~\bibnamefont
  {Weingartshofer}}}\ (\bibinfo  {publisher} {Kluwer Academic Publishers},\
  \bibinfo {year} {1991})\BibitemShut {NoStop}%
\bibitem [{\citenamefont {Barut}(1992)}]{Barut1992}%
  \BibitemOpen
  \bibfield  {author} {\bibinfo {author} {\bibfnamefont {A.~O.}\ \bibnamefont
  {Barut}},\ }\bibfield  {title} {\enquote {\bibinfo {title} {Nonlinear
  nonlocal classical field theory of quantum phenomena},}\ }\href@noop {}
  {\bibfield  {journal} {\bibinfo  {journal} {Int. J. Eng. Sci.}\ }\textbf
  {\bibinfo {volume} {30}},\ \bibinfo {pages} {1469--1473} (\bibinfo {year}
  {1992})}\BibitemShut {NoStop}%
\bibitem [{\citenamefont {Ellis}(2009)}]{ellis}%
  \BibitemOpen
  \bibfield  {author} {\bibinfo {author} {\bibfnamefont {R.}~\bibnamefont
  {Ellis}},\ }\href@noop {} {\enquote {\bibinfo {title} {The theory of large
  deviations},}\ } (\bibinfo {year} {2009}),\ \bibinfo {note} {see page 72 and
  citation there}\BibitemShut {NoStop}%
\end{thebibliography}%
		
	\end{document}